\documentclass[11pt]{elsarticle}

\usepackage{CJK,fancybox,color,graphicx,amsfonts,amscd,amsmath,amsthm,amssymb,enumerate, subfigure}
\usepackage{latexsym,epsfig,changebar,indentfirst}

\usepackage{pict2e}

\newtheorem{thm}{Theorem}                

\newtheorem{proposition}[thm]{Proposition}  
\newtheorem{Remark}{Remark} 
\renewcommand{\figurename}{Fig.}            

\journal{Applied Mathematics and Computation}

\begin{document}

\begin{frontmatter}



\title{  A well-balanced finite difference WENO scheme \\for shallow water flow model }


\author[label2]{Gang Li\corref{cor1}}
\ead{gangli1978@163.com}%
\author[label3]{Valerio Caleffi}
\ead{valerio.caleffi@unife.it}
\author[label4]{Zhengkun Qi}
\ead{qzkun@qdu.edu.cn}

\address[label2]{School of Mathematical Sciences, Qingdao University, Qingdao,
Shandong 266071, P.R. China.}

\address[label3]{Dipartimento di Ingegneria, Universit$\acute{a}$ degli Studi di Ferrara, Via G. Saragat, 1, 44122 Ferrara, Italy.}

\address[label4]{International College, Qingdao University, Qingdao,
Shandong 266071, P.R. China.}

\cortext[cor1]{Corresponding author. Tel.: +86 0532 85953660. Fax: +86 0532 85953660.}

\begin{abstract}
In this paper, we are concerned with shallow water flow model over non-flat bottom topography by high-order schemes. Most of the numerical schemes in the literature are developed from the original mathematical model of the shallow water flow. The novel contribution of this study consists in designing an finite difference weighted essentially non-oscillatory (WENO) scheme based on the alternative formulation of the shallow water
flow model, denoted as ``pre-balanced" shallow water equations and introduced
in \{Journal of Computational Physics 192 (2003) 422-451\}. This formulation greatly simplifies
the achievement of the well-balancing of the present scheme. Rigorous numerical analysis as well as extensive numerical results all verify that the current scheme preserves the exact conservation property. It is important to note that this resulting scheme also maintains the non-oscillatory property near discontinuities and keep high-order accuracy for smooth solutions at the same time.
\end{abstract}

\begin{keyword}
Shallow water flow model; finite difference WENO scheme; Source term; Exact conservation property; High-order accuracy
\end{keyword}

\end{frontmatter}




\section{Introduction}   \label{introduction}
\setcounter{equation}{0}
\setcounter{figure}{0}
\setcounter{table}{0}

In this paper, we are interested in numerical simulation for the shallow water flow model by high order finite difference schemes. The governing equations, referred to as the shallow water equations, have a wide applications in hydraulic and coastal engineering \cite{Bradford2002,Gottardi2004}. Some effects, such as friction on the bottom topography, wind forces, as well as variations of the channel width, can result in additional source terms to the governing equations. Herein, we only consider the geometrical source term due to the non-flat bottom topography. The one-dimensional case has the following form
\begin{equation}  \label{1D-SWE-original}
  \left\{
  \begin{array}{l}
  h_t + (hu)_x = 0,\\
  (hu)_t + \left(hu^2 + \frac{1}{2}gh^2\right)_x = -ghb_x,
  \end{array}
  \right.
\end{equation}
where $h$ and $u$ are the water depth and the depth-averaged water velocity, respectively, $g$ is the gravitational constant, and $b$ stands for the bottom topography. Therefore, $H=h+b$ denotes the water surface level and $hu$ represents the water discharge.

Shallow water equations with source terms are also called as \textit{balance laws}. Balance laws often admit steady state solution in which the source term is exactly balanced by the non-zero flux gradient. Thus it is desirable to maintain the balance at the discrete level, but such balance are usually neither a constant nor a polynomial function. For most numerical schemes, the truncation error between the flux gradient and the source term is not exactly zero for the above balance. So the numerical schemes for the shallow water equations with the source term is a challenging task with the presence of the source term. There are a lot of efforts focused on this subject in the literature, see e.g., \cite{Vazquez-Cendon, Jin, Kurganov}. Bermudez and Vazquez \cite{Bermudez}
firstly proposed the idea of ``exact conservation property'' (exact
C-property), which means that a scheme is exactly compatible with
the still water stationary solution
\begin{equation}    \label{stationary}
h+b = \mathrm{constant} \;\; and  \;\; u = 0.
\end{equation}
This property is necessary for the balance between the flux gradient
and the source term and is also known as \emph{well balancing}. An efficient
scheme should satisfies this property. Such schemes are often
regarded as well-balanced schemes after the pioneering works of Greenberg et al. \cite{Greenberg1996,Greenberg1997} and the study for the well-balanced schemes is currently a very active subject of research. The important advantages of well-balanced schemes over non-well-balanced schemes is that they can accurately resolve small perturbations of such steady state solution with relatively coarse meshes \cite{Noelle2010, Xing2011-1}. However, straightforward treatments to the source term can not preserve the exact C-property and even leads to unacceptable numerical results such as spurious numerical
oscillations. To construct well-balanced schemes, there are many attempts to handle the geometric source term. LeVeque \cite{LeVeque} brought forward a high-resolution Godunov-type finite volume scheme by a quasi-steady wave-propagation algorithm. Zhou {\em et al.} \cite{Zhou} designed a robust well-balanced scheme based on a Godunov-type method and a surface gradient method for the data reconstruction. By the aid of a special decomposition of the source term, Xing and Shu \cite{Xing2005} developed a high order well-balanced finite difference weighted non-oscillatory (WENO) schemes. Furthermore, Xing and Shu \cite{Xing-JSC,Xing-2006-JCP,Xing-2006-CICP} extended their idea and designed well-balanced finite volume WENO schemes and discontinuous Galerkin (DG) finite element methods for a class of balance laws equations. Audusse {\em et al.} \cite{Audusse2004, Audusse2005} designed a fast well-balanced scheme by a hydrostatic reconstruction procedure. With the help of a new quadrature formula for the source term, Noelle {\em et al.} \cite{Noelle2006} designed well-balanced finite volume WENO schemes with arbitrary order of accuracy.  Caleffi \cite{Caleffi2011} for the first time extended the Hermite WENO schemes \cite{Qiu2003} to the shallow water equations and obtained a well-balanced schemes even with the source term. More recently, Hou {\em et al.} \cite{Hou2013} proposed a robust well-balanced well-balanced cell-centered finite volume method on unstructured grids.
More information about well-balanced schemes can be found in the lecture note \cite{Noelle2010}.

Most of the numerical schemes reviewed above are developed considering as a
starting mathematical model the shallow water equations in their classical form.
An alternative formulation of the shallow
flow model, denoted as ``pre-balanced'' shallow water equations, was introduced by Roges {\em et al.} in \cite{Roges} to the aim of simplify the achievement of the well-balancing. The new formulation was obtained assuming as dependent variable the water elevation instead of the water depth and by a simple analytical manipulation. This formulation was adopted to develop different numerical schemes. For example in \cite{Canestrelli2009} the pre-balanced shallow water equations are integrated combining the PRICE-C method with a path-conservative method while in \cite{Li-2014} the same equations are integrated by a well-balanced central WENO scheme.

The original contribution of this research consists in an upwind finite difference
scheme based on the pre-balanced shallow water flow model \cite{Roges}. In analogy to \cite{Roges}, this new formulation greatly simplifies the achievement of exact C-property of
the present scheme. Rigorous numerical analysis as well as extensive benchmark
examples all verify the satisfaction of the exact C-property of the resulting scheme.
In addition, the high-order accuracy is obviously obtained at the same time.

This paper is organized as follows: in Section \ref{WENO-review}, we give a brief review of the finite difference WENO schemes. We propose a high order well-balanced finite difference WENO scheme in Section \ref{well-balanced-WENO} based on an equivalent governing equations for the shallow water flow model. Numerical experiments of one- and two-dimensional cases are carried out in Section \ref{1D-results} and Section \ref{2D-results}, respectively. Conclusions are given in Section \ref{conclusion}.


\section{A review of finite difference WENO schemes}\label{WENO-review}
The first finite difference WENO scheme was designed in 1996 by Jiang and Shu \cite{Jiang} for one- and two-dimensional hyperbolic conservation laws. More detailed information of WENO schemes can be found in the lecture note \cite{Shu1997}. For the latest advances regarding WENO schemes, we refer to the review \cite{Shu2009}. We begin with the description for the one-dimensional scalar conservation laws
 \begin{equation}\label{1D-laws}
 u_t+f(u)_x=0.
\end{equation}
For simplicity, we assume that the grid points $\{x_j \}$ are
uniform. We define the cell size and cells by $\Delta x = x_{j+1}-x_j$ and $I_j=[ x_{j-1/2},x_{j+1/2}]$ with $x_{j+1/2}=x_j+\Delta x/2$, respectively. A semidiscrete
conservative high order finite difference scheme of ($\ref{1D-laws}$)
can be formulated as follows
\begin{equation}\label{semi}
\frac{{\rm d} u_j(t)}{{\rm d}t}=-\frac{1}{\Delta
x}\left(\hat{f}_{j+1/2}-\hat{f}_{j-1/2}\right),
\end{equation}
where $u_j(t)$ is the numerical approximation to the point value
$u(x_j,t)$, and the numerical flux $\hat{f}_{j+1/2}$ is used to approximate
$h_{j+1/2}=h\left(x_{j+1/2}\right)$ with high order accuracy. Here $h(x)$
is implicitly defined as in \cite{Jiang}
$$
f(u(x))=\frac{1}{\Delta x}\int^{x+\Delta x/2}_{x-\Delta
x/2}h(\xi)d\xi.
$$

We take upwinding into account to maintain the numerical stability and split a general flux into two
parts either globally or locally
$$
f(u)=f^+(u)+f^-(u),
$$
where $ \displaystyle  \frac{ {\rm d} f^+(u) }{{\rm d} u} \geq 0$ and $ \displaystyle  \frac{ {\rm d} f^-(u) }{{\rm d} u} \leq 0$. With respect to $f^+(u)$ and $f^-(u)$, we can get numerical fluxes $\hat{f}^+_{j+1/2}$ and $\hat{f}^-_{j+1/2}$ using the WENO reconstruction, respectively. The computation of $\hat{f}^{+}_{j+1/2}$ and $\hat{f}^{-}_{j+1/2}$ is described in \ref{Appendix-A}. Finally, we get the numerical fluxes as follows $$\hat{f}_{j+1/2} = \hat{f}^{+}_{j+1/2} + \hat{f}^{-}_{j+1/2}.$$

With numerical fluxes $\hat{f}_{j+1/2}$ at hand, we can write the semidiscrete scheme (\ref{semi}) as an ordinary differential equation (ODE) system
$$
u_t = L(u).
$$
Eventually, we discretize this ODE system in time by the third-order total variation diminishing (TVD)
Runge-Kutta method \cite{Shu1988}
\begin{equation}  \label{rk3}
\begin{array} {lcl}
     u^{(1)} & = & u^n + \Delta t L\left(u^n\right),   \\
     u^{(2)} & = & \frac{3}{4} u^n + \frac{1}{4} u^{(1)} + \frac{1}{4} \Delta t L\left(u^{(1)}\right), \\
     u^{n+1} & = & \frac{1}{3} u^n + \frac{2}{3} u^{(2)} +   \frac{2}{3} \Delta t L\left(u^{(2)}\right).
\end{array}
\end{equation}

\section{The numerical scheme} \label{well-balanced-WENO}
In this section, we design a well-balanced finite difference WENO scheme for the shallow water flow model. For the sake of simplicity, we take the one-dimensional case as an example.

Herein, we apply the equivalent governing equations as in \cite{Roges}, where the water surface level $H$ instead of the
water depth $h$ as an unknown variable since the bottom topography b is independent of the time $t$ and the physical fluxes have been changed accordingly. Similar procedure is also used in the surface gradient method by Zhou {\em et al.} \cite{Zhou}, the centered scheme of Canestrelli {\em et al.} \cite{Canestrelli2009}, the discontinuous Galerkin finite element methods of Kesserwani and Liang \cite{Kesserwani2010}, and the central WENO scheme of Li {\em et al.} \cite{Li-2014}.

Equivalent to the original governing equations (\ref{1D-SWE-original}), we apply the following ones as in \cite{Roges},
 \begin{equation} \label{1D-SWE-modified}
\left\{
\begin{array}{l}
 H_t   + (hu)_x                     = 0,\\
(hu)_t + \left( \displaystyle \frac{   (hu)^2}{H-b} + \frac{1}{2}gH^2 - gHb   \right)_x = -gH b_x,
\end{array}
\right.
\end{equation}
which can be denoted by a compact vector form
$$U_t + f(U)_x = S,$$
where $U = (H, \; hu)^T, \; f(U) = \left(\displaystyle   hu, \; \frac{  (hu)^2}{H-b} + \frac{1}{2}gH^2 - gHb \right)^T, \; \mbox{and} \; S = (0, \; -g H b_x)^T.$

In this paper, we are intent on solving the equivalent governing equations (\ref{1D-SWE-modified}) by \emph{linear} schemes. For a given linear scheme, all the spatial derivatives are approximated by a \emph{linear} finite difference operator $D$ that satisfyies
\begin{equation} \label{linear-operator}
D(\alpha f_1 + \beta f_2)=\alpha D(f_1)+\beta D(f_2)
\end{equation}
for constants $\alpha, \; \beta$ and any grid functions $f_1, \; f_2$. For such linear schemes, we have

\begin{proposition}   \label{proposition-linear}
For the still water stationary solution (\ref{stationary}), linear schemes satisfying (\ref{linear-operator}) for the shallow water equations (\ref{1D-SWE-modified}) maintain the exact C-property.
\end{proposition}

\begin{proof}
 For still water stationary solution (\ref{stationary}), any consistent linear schemes satisfying (\ref{linear-operator}) are exact for the first equation
$(hu)_x = 0$ due to $u=0$. For the second one, with a given linear finite difference operator $D$, the error between the flux gradient and the source term reduces to
$$
\begin{array} {l}
D \left( \displaystyle   \frac{(hu)^2}{H-b} + \frac{1}{2}gH^2 - gHb    \right) + g H D(b) \\
= D\left(\frac{1}{2} gH^2 - gHb  + g H b \right)   \\
= D\left(\frac{1}{2}g H^2\right)  \\
\equiv 0,
\end{array}
$$
here the first equality is due to the linearity of $D$, $u=0$ and $H = h + b = constant$; the second
equality is just a simple algebra operation inside the parenthesis, and the last one is again due to
the fact that $H = h + b = constant$ and the consistency of the operator $D$. This completes the proof.
\end{proof}

Unfortunately, the finite difference WENO schemes described in Section \ref{WENO-review} are nonlinear. The
nonlinearity comes from the nonlinear weights, which in turn come
from the nonlinearity of the smoothness indicators.
In order to construct a linear scheme
which can maintain the exact C-property even with the presence of the nonlinearity of the
nonlinear weight, we adopt the following procedures. The resulting scheme maintains the exact C-property and the accuracy is not affected.

Firstly, we compute the numerical flux $\hat{f}_{j+1/2}$ to approximate the flux gradient $f(U)_x$. For similarity, we consider a finite difference WENO scheme with a global Lax-Friedrichs flux splitting, denoted by WENO-LF scheme. Now the physical flux $f(U)$ is written out as
$$f(U) = f^+(U) + f^-(U),$$
where
\begin{equation}\label{flux-splitting}
f^{\pm}(U)=\frac{1}{2}\left[ \left(\begin{array}{c}
hu \\
\displaystyle  \frac{(hu)^2}{H-b} + \frac{1}{2}gH^2 - gHb
\end{array}
\right) \pm \alpha _i\left(
\begin{array}{c}
H\\
hu
\end{array}
\right) \right],
\end{equation}
with
\begin{equation}\label{alpha}
\alpha_i = \max\limits_u  \big| \lambda_i(u) \big|,
\end{equation}
here $\lambda _i(u)$ being the $i$th eigenvalue of the Jacobian matrix $f'(U)$.

Moreover, in order to achieve better
numerical results at the price of more complicated computations, the WENO
reconstruction is always accompanied by a local characteristic
decomposition procedure \cite{Shu1997}, which is more robust than a component-wise version.

Subsequently, we will verify that for the still water stationary solution (\ref{stationary}), the present WENO scheme is a linear scheme. We refer to \ref{Appendix-B} for the complete verification of the linearity of the present scheme. Consequently, we have the following result according to the Proposition \ref{proposition-linear}.

\begin{proposition}
For the still water stationary solution (\ref{stationary}), the WENO-LF scheme as stated above for the shallow water equations (\ref{1D-SWE-modified}) maintains the exact C-property and their original accuracy.
\end{proposition}


\section{Numerical results} \label{results}
In this Section, we carry out extensive one- and two-dimensional numerical experiments to demonstrate the
performances of a fifth-order ($r=2$) finite difference WENO scheme. In all the numerical examples, time discretization is by the classical third-order Runge-Kutta method \cite{Shu1988}. The $CFL$ number is taken as $0.6$, except for the accuracy tests where smaller time step is taken to ensure that spatial errors dominate. The gravitation constant g is taken as $9.812 \,  \mbox{m}/\mbox{s}^2$.

\subsection{One-dimensional cases} \label{1D-results}
Firstly, we present numerical results of our fifth-order finite difference WENO-LF scheme for the one-dimensional model (\ref{1D-SWE-modified}).

\subsubsection{Testing the exact C-property}

 We verify the exact C-property of the resulting scheme by the following test cases in \cite{Xing2005} over two different bottom topographies on a computational domain
$[0,10]$. The first bottom topography is smooth
$$
\displaystyle b(x) = 5e^{-\frac{2}{5}(x-5)^2},
$$
and the second one is discontinuous
$$
b(x)= \left\{
\begin{array}{ll}
4 & \mathrm{if} \; 4 \leq x \leq  8, \\
0 & \mathrm{otherwise}.
\end{array}
\right.
$$
The initial data is the still water stationary solutions
$$h + b = 10  \; \;  \mbox{and} \; \;  u=0. $$
 We compute the solution up to $t = 0.5$ s on a mesh with $200$ uniform cells. In order to show
that the exact C-property is maintained even with round off error,
we apply single, double and quadruple precisions, respectively, to carry out the
computation. We present the $L^1$ and $L^{\infty}$ error for $h$ and
$hu$ in Tables \ref{t:smooth} and \ref{t:discontinuous} for the two different bottom
topographies. We can clearly observe that the $L^1$ and $L^{\infty}$
errors are all at the level of round off error for different precisions, and
verify the expected exact C-property accordingly.

 \subsubsection{  Testing the orders of accuracy}  \label{order}

In this test case, we test the fifth-order accuracy of the
scheme for the smooth solution. We apply the following bottom
topography and the initial condition
$$
b(x) = \sin^2(\pi x), \; h(x,0) = 5+e^{\cos(2\pi x)}, \; (hu)(x,0) = \sin(\cos(2\pi x)), \; x \in [0,1]
$$
with periodic boundary conditions. We compute the test case up to $t=0.1$ s and apply the
same fifth-order WENO scheme with $6400$ cells to obtain a
reference solution. In Table \ref{t:order-1D}, we
list the $L^1$ errors and orders of accuracy for $h$ and $hu$.
It is clear that we get the expected fifth-order accuracy for this
test case. Due to the space limitation, we do not present the $L^{\infty}$ errors and the orders of accuracy, since they are similar with the $L^1$ errors and the orders of accuracy.

\subsubsection{A small perturbation of a steady state water flow}

 The following quasi-stationary test case was proposed by LeVeque
\cite{LeVeque}. It is chosen to demonstrate the capability of the present scheme for the computation on a rapidly
varying flow over a smooth bottom topography and the perturbation of a
stationary state flow. The bottom topography consists of a bump
$$b(x)=
\left\{\begin{array}{ll}
0.25\left(\cos(10\pi(x-1.5))+1\right)  & \mbox{if} \; 1.4  \leq  x \leq  1.6,\\
0                           & \mathrm{otherwise},
\end{array}\right.
$$
and the initial condition is given as
$$ h(x,0)=\left\{\begin{array}{ll}
1-b(x)+\epsilon  & \mbox{if} \; 1.1  \leq  x  \leq  1.2, \\
1-b(x)           & \mathrm{otherwise},\end{array}\right. \; \mathrm{and} \;
u(x,0) = 0.
$$
where $\epsilon$ is a non-zero perturbation constant. Two
cases are considered: $\epsilon=0.2$ m (big pulse) and $\epsilon=0.001$ m (small pulse).

We present the water surface level $h+b$ and the water discharge $hu$ at
$t = 0.2$ s against reference solutions in Figs. \ref{f:big-5} and \ref{f:small-5} for the big pulse and the small
pulse cases, respectively. The numerical results are resolved accurately, free of spurious numerical oscillations, and look very comparable to those found in the other existing literature.

 \subsubsection{Steady flow over a hump}

In this example, we employ three established benchmark test cases \cite{Vazquez-Cendon} with different boundary conditions. These test cases involve transcritical, supercritical and subcritical flows, respectively, in a $25$ m channel over a bump
$$b(x)=
\left\{
\begin{array}{ll}
0.2 - 0.05(x-10)^2       & \mathrm{if} \;  8 \leq x \leq  12,\\
0    & \mathrm{otherwise}.
\end{array}
\right.
$$
The initial data are defined by
$$h(x,0) = 0.33  \;\;  \mbox{and} \; \; u(x,0) = 0.$$

We employ same computational parameters for the following three cases: uniform mesh with $200$ cells, final time $t = 200 $ s. Exact solutions for the three cases can be found in \cite{Goutal1997}.

 Case 1:  Transcritical flow without a shock

A unit discharge of $1.53$ $\mbox{m}^2$/s is imposed at the upstream boundary,
and the open boundary conditions ($du/dx = 0$) are applied
at the downstream one. We present the water surface level $h + b$ and
 the water discharge $hu$ in \figurename \, \ref{f:bump-a}. It is obvious that the numerical solutions are very good agreement with the exact ones.

Case 2:  Transcritical flow with a shock

A unit discharge of $0.18$ $\mbox{m}^2/\mbox{s}$  is imposed on the upstream
boundary and a depth of $0.33 $ m  is imposed on the downstream
boundary. We show the water surface level $h + b$ and the water discharge $hu$
 against exact solutions in \figurename \, \ref{f:bump-b}. The numerical results are free of spurious oscillations, which verifies the essentially non-oscillatory property of the current scheme.

Case 3: Subcritical flow

A unit discharge of $4.42$ $\mbox{m}^2/\mbox{s}$ is imposed on the upstream boundary and a depth of
$2$ m is imposed on the downstream boundary. The numerical results are compared with exact
solutions in \figurename \, \ref{f:bump-c}, and very good agreement
is achieved.

\subsubsection{ The dam break problem over a rectangular bump}

 This test case was used in \cite{Vukovic}. Herein, we simulate a dam break problem over a rectangular bump, which involves a rapidly varying water flow over a discontinuous bottom topography. The bottom topography contains a rectangular bump:
$$
b(x)=
\left\{\begin{array}{ll}
8  & \mbox{if} \; \big|x-750\big| \leq  1500/8,\\
0  & \mathrm{otherwise}.
\end{array}\right. $$
The initial conditions are given as follows
$$ h(x,0)
=\left\{\begin{array}{ll}
20 - b(x)   & \mbox{if} \; x \leq  750, \\
15 - b(x)   & \mathrm{otherwise},
\end{array}
\right.
\;
\mathrm{and}
\;\;
u(x,0) = 0.
$$

We present numerical results against reference solutions in
\figurename \, \ref{rectangular-bump}, which indicate that the numerical results keep the essentially non-oscillatory property and are in good agreement with the reference solutions.

\subsubsection{ The tidal wave flow}

This example was used in \cite{Bermudez1994}, in which almost exact solutions (a very good asymptotically derived
approximation) were given. The bottom topography is defined as
$$b(x) = 10 + \frac{40 x}{L} + 10 \sin\left( \pi \left(\frac{40 x}{L} - \frac{1}{2}\right)\right),$$
with $L=14,000$ m being the channel length. Herein, we take the following initial conditions
$$h(x,0) = 60.5 - b(x),  \; \;  hu(x,0) = 0,$$
and boundary conditions
$$h(0,t) = 64.5 - 4 \sin\left(\pi \left( \frac{4 t}{86,400} + \frac{1}{2}\right)\right), \; \;  hu(L, t) = 0.$$
By means of the asymptotic analysis in \cite{Bermudez1994}, we can obtain the following almost exact solutions
$$h(x,t) = 64.5 - b(x) - 4 \sin\left(\pi \left( \frac{4 t}{86,400} + \frac{1}{2}\right)\right), $$
$$hu(x,t) = \frac{(x-L)\pi}{5400} \cos\left(\pi \left( \frac{4 t}{86,400} + \frac{1}{2}\right)\right).$$

We compute the example on a mesh with $200$  uniform cells up to $t = 7552.13$ s and present numerical results against exact solutions in
\figurename \, \ref{tidal-flow}, which strongly suggests that the numerical results are in good agreement with the exact ones.

\subsubsection{  $1$-rarefaction and $2$-shock problem}  \label{1-rarefaction-2-shock}

Then we consider a example over a step bottom topography \cite{Alcrudo} to further test our scheme. The bottom topography consists of a step
$$b(x)=
\left\{
\begin{array}{ll}
0    & \mathrm{if} \;    x \leq  0,\\
1    & \mathrm{otherwise},
\end{array}
\right.
$$
on a computational domain $[-10,10]$, and the initial data are as follows
$$ h(x,0)=\left\{\begin{array}{ll}
4  & \mbox{if} \; x \leq  0,\\
1  & \mbox{otherwise},
\end{array}
\right.
\;  \;
u(x,0)=0. $$

This test case produces a $1$-rarefaction spreading to the left and a
$2$-shock traveling to the right. We illustrate the water surface level $h+b$ and the water discharge $hu$ at $t = 1$ s
against exact solutions in \figurename \, \ref{f:1D-step-RS}. We can clearly observe that the numerical results keep a sharp discontinuity transition.

\subsubsection{ $1$-shock and $2$-shock problem}

This test case is also over the same step bottom topography as in Section \ref{1-rarefaction-2-shock} on a computational domain $[-10,10]$. The initial data are given by
$$ h(x,0)=
\left\{
\begin{array}{ll}
4  & \mbox{if} \; x \leq  0,\\
1  & \mathrm{otherwise},
\end{array}\right.
\; \mathrm{and} \;\;
u(x,0)=\left\{
\begin{array}{ll}
5     & \mathrm{if} \; x \leq  0,\\
-0.9  & \mathrm{otherwise}.
\end{array}
\right. $$

This test case produces two shocks: the first one moving to the left
and the second one to the right. We present the water surface level $h+b$
and the water discharge $hu$ at $t = 1$ s against exact solutions in \figurename \, \ref{f:1D-step-SS}. It is evident that the numerical results possess a good resolution and are almost free of spurious numerical oscillations.

\begin{Remark}
 As can be seen in Figs. \ref{f:1D-step-RS} and \ref{f:1D-step-SS}, there are some minor numerical oscillations for the water discharge. The occurrence is mainly due to the non-flat bottom topography which can not be handled since the imbalance between the flux gradient and the source term as well as the moving water flow \cite{Benkhaldoun2007}. As the correct capturing of the water discharge is more difficult than the water surface level, so we are satisfied with the present numerical results. Although the current numerical scheme is well-balanced for the still water stationary solution, it is
not able to maintain such a desirable exact C-property for the moving steady-state problem. Therefore, from Figs. \ref{f:1D-step-RS} and \ref{f:1D-step-SS}, we can observe disturbance to discharge in those areas with abrupt change of the bottom topography and the water depth, which is a common phenomenon also predicted by other well-balanced schemes (e.g., \cite{Zhou,Roges}). The disturbance is then advected by the flow as a wave and reaches the location as indicated in Figs. \ref{f:1D-step-RS} and \ref{f:1D-step-SS}.
\end{Remark}

\subsection{Two-dimensional cases } \label{2D-results}

Subsequently, we consider two-dimensional cases. In analogy with the one-dimensional case, the governing equations are as follows
\begin{equation}\label{2D-SWE}
\left\{
\begin{array}{l}
 H_t   + (hu)_x + (hv)_y                    = 0,\\
(hu)_t + \left( \displaystyle  \frac{(hu)^2}{H-b} + \frac{1}{2}gH^2 - gHb    \right)_x + (huv)_y = -gH b_x,\\
(hv)_t + (huv)_x + \left( \displaystyle  \frac{(hv)^2}{H-b} + \frac{1}{2}gH^2 - gHb \right)_y   = -gH b_y,
\end{array}
\right.
\end{equation}
where $v$ denotes the $y$-direction velocity, and the remaining notations are the same as in the one-dimensional case.

\subsubsection{ Testing the exact C-property }

We apply this test case to demonstrate the fact that for the two-dimensional case the present scheme indeed maintains the exact C-property over a non-flat bottom topography
$$\displaystyle
b(x,y)=0.8e^{ -50((x-0.9)^2+(y-0.5)^2))}, \; (x,y) \in [0,1] \times [0,1].
$$
The initial data are given by
$$
h(x,y,0) = 1 - b(x,y), \; \; u(x,y,0) = v(x,y,0) = 0.
$$

We compute the example up to $t = 0.1$ s on a mesh with $100 \times 100$ cells. We apply single, double and quadruple precisions, respectively, to carry out the
computation. We present the $L^1$ error for $h$, $hu$, and $hv$ in Table \ref{t:exact-c-property-2D}. We can clearly observe that the $L^1$
errors are at the level of round off error for different precisions, and
verify the expected exact C-property accordingly.

\subsubsection{ Testing the orders of accuracy} \label{order-2D}

In this example, we test the numerical orders of accuracy when the resulting scheme is applied to the
following two-dimensional problem on a square domain $[0,1] \times [0,1]$ as in \cite{Xing2005}. We adopt the following bottom topography
$$\displaystyle
b(x,y)= \sin(2\pi x) + \cos(2\pi y),
$$
and the initial data
$$
(h, hu, hv)(x,y,0) =  \left(  10 + e^{\sin(2\pi x) } \cos(2 \pi y),\;  \sin(\cos(2\pi x)) \sin(2\pi y), \; \cos(2\pi x)\cos(\sin(2\pi y))   \right),
$$
with periodic boundary conditions.

We compute this test case up to $t=0.05$ s and apply the
same fifth-order WENO scheme with $1600 \times 1600$ cells to obtain reference solutions. In Table \ref{t:order-2D}, we
list the $L^1$ errors and orders of accuracy for $h$, $hu$ and $hv$.
It is obvious that we get the expected fifth-order accuracy for this test case. Due to the space limitation, we do not present the $L^{\infty}$ errors and the orders of accuracy, since they are similar with the $L^1$ errors and the orders of accuracy.

\subsubsection{ A small perturbation of a two-dimensional steady state water flow}

We consider the test case on a rectangular domain $[0,2]\times[0,1]$. The bottom
topography contains an isolated elliptical shaped hump
$$\displaystyle
b(x,y)=0.8e^{\displaystyle-5(x-0.9)^2-50(y-0.5)^2)},
$$
the initial condition are given by
$$h(x,y,0)=\left\{
\begin{array}{ll}
1-b(x,y)+0.01  &\mathrm{if} \, 0.05 \leq  x \leq  0.15, \\
1-b(x,y)       &\mathrm{otherwise},
\end{array}
\right. \mathrm{and}\;\; u(x,y,0) = v(x,y,0)=0.
$$

For comparison, we present contours of the water surface level $h+b$ on two different meshes with
$200 \times 100$ and $600 \times 300$ uniform cells in \figurename \, \ref{f:2D-perturbation}.
\figurename \, \ref{f:2D-perturbation} displays the right-going disturbance as
it propagates past the hump. The numerical results suggest that our
schemes can resolve complex small-scale features of the water flow very well. The numerical results are comparable with those in \cite{Xing2005}.

\section{Conclusions} \label{conclusion}
In this paper, we develop a well-balanced finite difference WENO scheme for the shallow water flow model based on equivalent governing equations. Rigorous numerical analysis as well as extensive numerical experiments all suggest that the present scheme maintains the exact C-property for the still water stationary solution. It is also important that the scheme obtains the expected high-order accuracy for smooth solutions, and keeps essentially non-oscillatory property near discontinuities. Based on the current governing equations, the research for the well-balanced finite volume WENO scheme are ongoing.

\section*{Acknowledgements}

The research of the first author is supported
by the National Natural Science Foundation of P.R. China (No. 11201254) and the Project for Scientific Plan of Higher Education in Shandong Providence of P.R. China (No. J12LI08).  The second author is sponsored by MIUR, Prin2009, Metodi numerici innovativi per problemi iperbolici con applicazioni in fluidodinamica, teoria cinetica e biologia computazionale. This work was partially performed at the State Key Laboratory of Science/Engineering Computing of China by virtue of the computational resources of Professor Li Yuan's group. The first author is also thankful to Professor Li Yuan for his kind invitation.


\appendix

\section{WENO reconstruction procedure for numerical fluxes }  \label{Appendix-A}

By means of the WENO reconstruction procedure, $\hat{f}^+_{j+1/2}$ can be expressed as
\cite{Jiang}
\begin{equation} \label{weno}
\hat{f}^+_{j+1/2}=\sum^{r}_{k=0}\omega_kq^r_k\left(f^+_{j+k-r},\ldots,f^+_{j+k}\right),
\end{equation}
with $\omega_k$ being a nonlinear weight,
$f^+_i=f^+(u_i),\,i=j-r,\ldots,j+r,$ and
\begin{equation}
q^r_k\left(\mbox{g}_0,\ldots,\mbox{g}_{r}\right)=\sum^{r}_{l=0}a^r_{k,l}\mbox{g}_l
\end{equation} are the low order reconstruction to $\hat{f}^+_{j+1/2}$ on the $k$th
stencil $S_k=(x_{j+k-r},\ldots,x_{j+k}),k=0,1,\ldots,r$, and
$a^r_{k,l}, 0 \leq  k, \, l \leq  r$ are constant coefficients, see
\cite{Shu1997} for more details.

The nonlinear weight $\omega_k$ in ($\ref{weno}$) satisfies
$ \sum\limits^{r}_{k=0}\omega_k=1, $
and is designed to yield $(2r+1)$th-order accuracy in smooth regions
of the solution. In \cite{Jiang,Shu1997}, the nonlinear weight
$\omega_k$ is formulated as
\begin{equation}\label{linearweight}
\omega_k=\frac{\alpha_k}{
\sum\limits^{r}_{l=0}\alpha_l},\;\mbox{with}\;
\alpha_k=\frac{C^r_k}{\left(\varepsilon+IS_k\right)^2},\,k=0,1,\ldots,r,
\end{equation}
where $C^r_k$ is the linear weight. $IS_k$ is a smoothness indicator of $f^+(u)$ to measure the
smoothness of $f^+(u)$ on the stencil $S_k, \; k=0,1,\ldots,r$, and $\varepsilon$ is a small constant
used here to avoid the denominator becoming zero, we take
$\varepsilon=10^{-6}$ for all test cases in this paper. We employed
the smoothness indicators proposed in \cite{Jiang,Shu1997}, i.e.,
$$IS_k=\sum^{r}_{l=1}\int^{x_{j+1/2}}_{x_{j-1/2}}(\Delta
x)^{2l-1}\left(q_k^{(l)}\right)^2dx,$$
where $q^{(l)}_k$ is the
$l$th-derivative of $q_k(x)$ and $q_k(x)$ is the reconstruction
polynomial of $f^+(u)$ on stencil $S_k$ such that $$\frac1{\Delta
x} \int_{I_i} q_k(x)dx=f^+_i,\,i=j+k-r,\ldots,j+k.$$

The procedure for the reconstruction of $\hat{f}^-_{j+\frac{1}{2}}$
is a mirror symmetry to that of $\hat{f}^+_{j+1/2}$ with respect to the grid point $x_{j+1/2}$, so we will not present it here to save space.


\section{ Verification of the linearity of the scheme} \label{Appendix-B}

By virtue of the WENO reconstruction procedure, $\hat{f}^{+}_{j+1/2}$ can be written out in the following form
\begin{equation} \label{fp}
\begin{array}{lcl}
\hat{f}^{+}_{j+1/2}
&=& \sum\limits_{k=-r}^{r}c_kf^+_{j+k}\\
&=& \sum\limits_{k=-r}^{r}c_k \left[\frac{1}{2}\left(f_{j+k} + \alpha U_{j+k} \right) \right] \\
&=& \frac{1}{2}\sum\limits_{k=-r}^{r}c_kf_{j+k} + \frac{1}{2} \sum\limits_{k=-r}^{r}c_k\left(\alpha U_{j+k}\right),
\end{array}
\end{equation}
where $ f^{+}=f^{+}(U) $ as in (\ref{flux-splitting}), $c_k$ is a $2\times2$ matrix depending nonlinearly
on the smoothness indicators of $f^{+}$ on the stencil
$\{x_{j-r},\ldots,x_{j+r}\}$, and $\alpha$ is a $2\times2$ diagonal
matrix involving $\alpha_i$ in (\ref{alpha}).

Analogously, we can write $\hat{f}^{-}_{j+1/2}$ as follows
\begin{equation} \label{fm}
\begin{array}{lcl}
\hat{f}^{-}_{j+1/2}
&=& \sum\limits_{k=-r+1}^{r+1}a_kf^{-}_{j+k} \\
&=& \sum\limits_{k=-r+1}^{r+1}a_k \left[  \frac{1}{2}\left(f_{j+k} - \alpha U_{j+k} \right)\right] \\
&=& \frac{1}{2}\sum\limits_{k=-r+1}^{r+1}a_kf_{j+k} - \frac{1}{2}
\sum\limits_{k=-r+1}^{r+1}a_k\left(\alpha U_{j+k}\right),
\end{array}
\end{equation}
where $a_k$ is also a $2\times2$ matrix depending nonlinearly
on the smoothness indicators of $f^{-}$ on the stencil
$\{x_{j-r+1},\ldots,x_{j+r+1}\}$.

So we have
\begin{equation}
\hat{f}_{j+1/2}=\hat{f}_{j+1/2}^{+}+\hat{f}_{j+1/2}^{-}.
\end{equation}

Similarly, $\hat{f}_{j-1/2}^{+}$ has the following form
\begin{equation} \label{fp-}
\begin{array}{lcl}
\hat{f}_{j-1/2}^{+}
&=& \sum\limits_{k=-r-1}^{r-1}\hat{c}_kf_{j+k}^{+} \\
&=& \sum\limits_{k=-r-1}^{r-1}\hat{c}_k \left[  \frac{1}{2}\left(f_{j+k} + \alpha U_{j+k}\right)\right] \\
&=& \frac{1}{2}\sum\limits_{k=-r-1}^{r-1}\hat{c}_kf_{j+k} +
\frac{1}{2} \sum\limits_{k=-r-1}^{r-1}\hat{c}_k \left(\alpha U_{j+k}\right),
\end{array}
\end{equation}
where $\hat{c}_k$ is a $2\times2$ matrix depending nonlinearly on the
smoothness indicators of $f^{+}$ on the
stencil $\{x_{j-r-1},\ldots,x_{j+r-1}\}$.

Finally, $\hat{f}_{j-1/2}^{-}$ can be written out in the below form
\begin{equation}\label{fm-}
\begin{array}{lcl}
\hat{f}_{j-1/2}^{-} &=& \sum\limits_{k=-r}^{r} c_k f_{j-k}^{-} \\
                       &=& \sum\limits_{k=-r}^{r} c_k \left[ \frac{1}{2} \left(f_{j-k} -   \alpha U_{j-k}\right)\right]\\
                       &=& \frac{1}{2}\sum\limits_{k=-r}^{r}c_kf_{j-k} -
                       \frac{1}{2}\sum\limits_{k=-r}^{r}c_k \left(\alpha U_{j-k}\right),
\end{array}
\end{equation}
where $c_k$ is the same $2\times2$ matrix as in
(\ref{fp}).

Consequently, we can obtain
\begin{equation}\label{f-minus}
\hat{f}_{j-1/2} =\hat{f}_{j-1/2}^{+} +
\hat{f}_{j-1/2}^{-}.
\end{equation}

With the formulae in (\ref{fp}), (\ref{fm}), (\ref{fp-})
and (\ref{fm-}), the approximation to $f(U)_x$ can be eventually written out
as follows
\begin{equation}\label{approximation-f}
\begin{array}{lcl}
f\left(U\right)_x|_{x=x_j}
&\approx& \frac{1}{\Delta
x}\left(\hat{f}_{j+1/2} -
\hat{f}_{j-1/2}\right)\\
&=&\frac{1}{\Delta x}\left[ \left(
\frac{1}{2}\sum\limits_{k=-r}^{r}c_kf_{j+k} + \frac{1}{2}
\sum\limits_{k=-r}^{r}c_k(\alpha U_{j+k}) +
\frac{1}{2}\sum\limits_{k=-r+1}^{r+1}a_kf_{j+k} - \frac{1}{2}
\sum\limits_{k=-r+1}^{r+1}a_k\left(\alpha U_{j+k}\right)
 \right)\right. \\
&-& \left.\left(
\frac{1}{2}\sum\limits_{k=-r-1}^{r-1}\hat{c}_kf_{j+k} + \frac{1}{2}
\sum\limits_{k=-r-1}^{r-1}\hat{c}_k (\alpha U_{j+k}) +
\frac{1}{2}\sum\limits_{k=-r}^{r}c_kf_{j-k} -
                       \frac{1}{2}\sum\limits_{k=-r}^{r}c_k \left(\alpha U_{j-k}\right)\right)
\right]
 \\
&=& \frac{1}{2\Delta x}\left(\sum\limits_{k=-r}^{r}c_kf_{j+k} -
\sum\limits_{k=-r-1}^{r-1}\hat{c}_kf_{j+k}\right)\\
&+& \frac{1}{2\Delta x}\left(\sum\limits_{k=-r+1}^{r+1}a_kf_{j+k} -
\sum\limits_{k=-r}^{r}c_kf_{j-k}\right)\\
&+& \frac{1}{2\Delta x}\left(\sum\limits_{k=-r}^{r}c_k(\alpha
U_{j+k}) -
\sum\limits_{k=-r-1}^{r-1}\hat{c}_k(\alpha U_{j+k})\right)\\
&+& \frac{1}{2\Delta x}\left(\sum\limits_{k=-r}^{r}c_k(\alpha
U_{j-k}) - \sum\limits_{k=-r+1}^{r+1}a_k\left(\alpha U_{j+k}\right)\right) \\
&=& \mathbf{P1} +  \mathbf{P2} + \mathbf{P3} + \mathbf{P4}.
\end{array}
\end{equation}

Subsequently, we will verify that the formula $\frac{1}{\Delta
x}\left(\hat{f}_{j+1/2} - \hat{f}_{j-1/2}\right)$ for the approximation to $f\left(U\right)_x|_{x=x_j}$ is a finite difference operator. It should be noted that with
$\pm \alpha U= \pm \alpha
\left(
\begin{array}{c}
H \\
hu
\end{array}
\right)$
 in the flux splitting (\ref{flux-splitting}), this vector becomes a constant vector for the still water stationary solution (\ref{stationary}). By $U$ we denote $U_{j+k}$ with an
abuse of notation. So $\alpha U_{j+k} = \alpha U$ is also a constant
vector. Thus
\begin{equation}
\begin{array}{lcl}
\mathbf{P3} &=& \frac{1}{2\Delta x}\left(\sum\limits_{k=-r}^{r}c_k(\alpha U_{j+k})
- \sum\limits_{k=-r-1}^{r-1}\hat{c}_k(\alpha U_{j+k})\right) \\
&= &
\frac{1}{2\Delta x}\left(\sum\limits_{k=-r}^{r}c_k(\alpha U ) -
\sum\limits_{k=-r-1}^{r-1}\hat{c}_k(\alpha U )\right) \\
&= & \frac{1}{2\Delta
x}\left[\left(\sum\limits_{k=-r}^{r}c_k\right)(\alpha U ) -
\left(\sum\limits_{k=-r-1}^{r-1}\hat{c}_k\right)(\alpha U )\right]
\\
&= & \frac{1}{2\Delta x} \left[I \cdot(\alpha U) - I \cdot(\alpha U)\right] \\
&= & 0,
\end{array}
\end{equation}
where $I$ is a $2\times2$ identity matrix, the equivalents
$\sum\limits_{k=-r}^{r}c_k = I$ and
$\sum\limits_{k=-r-1}^{r-1}\hat{c}_k = I$ are due to the consistency
of the WENO reconstruction.

For the still water
stationary solution (\ref{stationary}), and with the similar
procedure as above, we can obtain that
\begin{equation}
\mathbf{P4} = \frac{1}{2\Delta x}\left(\sum\limits_{k=-r}^{r}c_k(\alpha U_{j-k}) -
\sum\limits_{k=-r+1}^{r+1}a_k(\alpha U_{j+k})\right) = 0.
\end{equation}

As a result, the approximation to $f(U)_x$ in
(\ref{approximation-f}) can be eventually written out as
\begin{equation}\label{Df}
\begin{array}{rcl}
f(U)_x|_{x=x_j}
&\approx& \frac{1}{\Delta x}(\hat{f}_{j+1/2} - \hat{f}_{j-1/2}) \\
&=& \mathbf{P1}  +  \mathbf{P2} \\
&=& \frac{1}{2\Delta x}\left(\sum\limits_{k=-r}^{r}c_k f_{j+k} -
\sum\limits_{k=-r-1}^{r-1}\hat{c}_kf_{j+k}\right)
+ \frac{1}{2\Delta x}\left(\sum\limits_{k=-r+1}^{r+1}a_kf_{j+k} -
\sum\limits_{k=-r}^{r}c_k f_{j-k}\right) \\
&=& \sum\limits_{k=-r-1}^{r+1}\beta_kf_{j+k} \\
&\triangleq& D_f(f)\big|_{x=x_j},
\end{array}
\end{equation}
where $D_f$ denotes a finite difference operator depending the flux $f(U)$ and $\beta_k$ is a $2\times2$ matrx depending on the smoothness indicators involving $f^+(U)$ and $f^-(U)$. The key idea of our scheme is to apply the operator $D_f$ in (\ref{Df}) with the fixed coefficient matrix $\beta_k$, to approximate the source term $(0, b)_x^T$. This amounts to split the source term as
\begin{equation}     \label{source-splitting}
\left(
\begin{array}{c}
0 \\
b
\end{array}
\right)_x
=
\frac{1}{2}\left(
\begin{array}{c}
0 \\
b
\end{array}
\right)_x + \frac{1}{2}\left(
\begin{array}{c}
0 \\
b
\end{array}
\right)_x,
\end{equation}
and apply the finite difference operator $D_f$ to approximate
them. Concretely speaking, one half part of the source term is approximated by
the operator $D_f$ with coefficients obtained from the computation of $f^+(U)$, and the remaining part by the operator $D_f$ with coefficients coming from the computation of $f^-(U)$.

A key observation is that the operator $D_f$ in (\ref{Df}) with the fixed coefficient matrices $\beta_k$ is a \emph{linear} finite difference operator on any grid function as in (\ref{linear-operator}). We thus conclude that for the still water stationary solution (\ref{stationary}), the present WENO scheme is a linear scheme, even with the global Lax-Friedrichs flux splitting (\ref{flux-splitting}) as well as the local characteristic
decomposition procedure.



\bibliographystyle{elsarticle-num}



\bibliographystyle{plain}

\newpage

\begin{table}
\centering
\caption{$L^1$ errors and orders of accuracy for the test
case in Section \ref{order}.}     \label{t:order-1D}
\begin{tabular}{c*{4}{c}}\hline
&\multicolumn{2}{c}{$h$}&\multicolumn{2}{c}{$hu$} \\ \cline{2-5}
\rule{0mm}{1.1em}\raisebox{2.5ex}[0ex]{N}&
\multicolumn{1}{c}{$L^{1} \; \mbox{error}$} &
\multicolumn{1}{c}{Order} & \multicolumn{1}{c}{$L^{1} \;
\mbox{error}$} & \multicolumn{1}{c}{Order}
\\\hline
  25 &   1.7486E-02 &          &   1.1294E-01 &            \\ \hline
    50 &   2.2133E-03 &     2.98 &   1.9663E-02 &     2.52 \\ \hline
   100 &   3.3157E-04 &     2.74 &   2.8131E-03 &     2.81 \\ \hline
   200 &   2.3391E-05 &     3.83 &   2.0167E-04 &     3.80 \\ \hline
   400 &   9.4357E-07 &     4.63 &   8.1928E-06 &     4.62 \\ \hline
   800 &   2.9898E-08 &     4.98 &   2.5426E-07 &     5.01 \\ \hline
\end{tabular}
\end{table}

\begin{table}
\centering
\caption{$L^1$ and $L^{\infty}$ error for different
precisions for the still water stationary solution over a smooth
bottom topography.}  \label{t:smooth}
\begin{tabular}{l*{4}{r}}\hline
& \multicolumn{2}{c}{$L^1 \,\,
\mbox{error}$}&\multicolumn{2}{c}{$L^{\infty} \,\,\mbox{error}$} \\\cline{2-5}
\rule{0mm}{1.1em}\raisebox{2.5ex}[0ex]{Precision}&
\multicolumn{1}{c}{$h$} & \multicolumn{1}{c}{$hu$} &
  \multicolumn{1}{c}{$h$} & \multicolumn{1}{c}{$hu$}    \\\hline
Single &
1.14E-06 & 1.612E-06 & 3.81E-06 & 5.23E-06  \\
\cline{1-5}
Double & 6.14E-16 & 4.12E-15 & 1.95E-15 & 1.48E-16  \\ \cline{1-5}
Quadruple & 1.57E-33 & 2.94E-32 &
6.98E-33 & 9.12E-32  \\ \cline{1-5}
\end{tabular}
\end{table}

\begin{table}
\centering
\caption{$L^1$ and $L^{\infty}$ error for different
precisions for the still water stationary solutions over a discontinuous
bottom topography.}  \label{t:discontinuous}
\begin{tabular}{l*{4}{r}}\hline
& \multicolumn{2}{c}{$L^1 \,\,
\mbox{error}$}&\multicolumn{2}{c}{$L^{\infty} \,\,\mbox{error}$}
\\\cline{2-5}
\rule{0mm}{1.1em}\raisebox{2.5ex}[0ex]{Precision}&
\multicolumn{1}{c}{$h$} & \multicolumn{1}{c}{$hu$} &
  \multicolumn{1}{c}{$h$} &
\multicolumn{1}{c}{$hu$}    \\\hline
Single &
1.53E-06 & 3.70E-07 & 1.91E-06 & 2.53E-06  \\
\cline{1-5}
Double & 4.35E-16 & 3.62E-15 & 1.60E-16 & 1.17E-15  \\ \cline{1-5}
Quadruple & 1.43E-33 & 2.15E-32 &
4.09E-33 & 5.64E-32  \\ \cline{1-5}
\end{tabular}
\end{table}

\begin{table}
\centering
\caption{$L^1$ for different
precisions for the still water stationary solutions over a non-flat bottom topography.}  \label{t:exact-c-property-2D}
\begin{tabular}{l*{3}{r}}\hline
& \multicolumn{3}{c}{$L^1 \,\, \mbox{error}$}  \\\cline{2-4}
\rule{0mm}{1.1em}\raisebox{2.5ex}[0ex]{Precision}&
\multicolumn{1}{c}{$h$} & \multicolumn{1}{c}{$hu$} &   \multicolumn{1}{c}{$hv$}    \\\hline
Single & 5.83E-08 & 2.91E-07 & 2.93E-07  \\ \hline
Double & 1.63E-16 & 6.43E-16 & 6.45E-16    \\ \hline
Quadruple & 2.13E-34 & 4.65E-34 & 4.39E-34    \\ \hline
\end{tabular}
\end{table} 

\begin{table}
\centering
\caption{$L^1$ errors and orders of accuracy for the test case in Section \ref{order-2D}.}     \label{t:order-2D}
\begin{tabular}{cc*{6}{c}}\hline
&\multicolumn{2}{c}{$h$}&\multicolumn{2}{c}{$hu$} &\multicolumn{2}{c}{$hv$}\\ \cline{3-8}
\rule{0mm}{1.1em}\raisebox{3.ex}[0ex]{$N_x \times N_y$}&\rule{0mm}{1.1em}\raisebox{3.ex}[0ex]{CFL}&
\multicolumn{1}{c}{$L^{1} \; \mbox{error}$} &
\multicolumn{1}{c}{Order} & \multicolumn{1}{c}{$L^{1} \;
\mbox{error}$} & \multicolumn{1}{c}{Order} & \multicolumn{1}{c}{$L^{1} \;
\mbox{error}$} & \multicolumn{1}{c}{Order} \\\hline
$25 \times 25$   & 0.6 &  1.1878E-002  &       &  3.6702E-002   &       & 9.8931E-002                   \\ \hline
$50 \times 50$   & 0.6 & 1.4841E-003  & 3.00  &  4.5263E-003   & 3.02  & 1.3532E-002   &  2.87         \\ \hline
$100 \times 100$ & 0.6 & 1.1262E-004  & 3.72  &  3.5071E-004   & 3.69  & 1.0558E-003   &  3.68         \\ \hline
$200 \times 200$ & 0.4 & 4.9428E-006  & 4.51  &  1.6844E-005   & 4.38  & 4.6660E-005   &  4.50         \\ \hline
$400 \times 400$ & 0.3 & 1.7866E-007  & 4.79  &  6.6166E-007   & 4.67  & 1.6866E-006   &  4.79         \\ \hline
$800 \times 800$ & 0.2 & 6.0255E-009  & 4.89  &  2.3751E-008   & 4.80  & 5.6882E-008   &  4.89         \\ \hline
\end{tabular}
\end{table}

\begin{figure}
\centering
\includegraphics[width=2.45in]{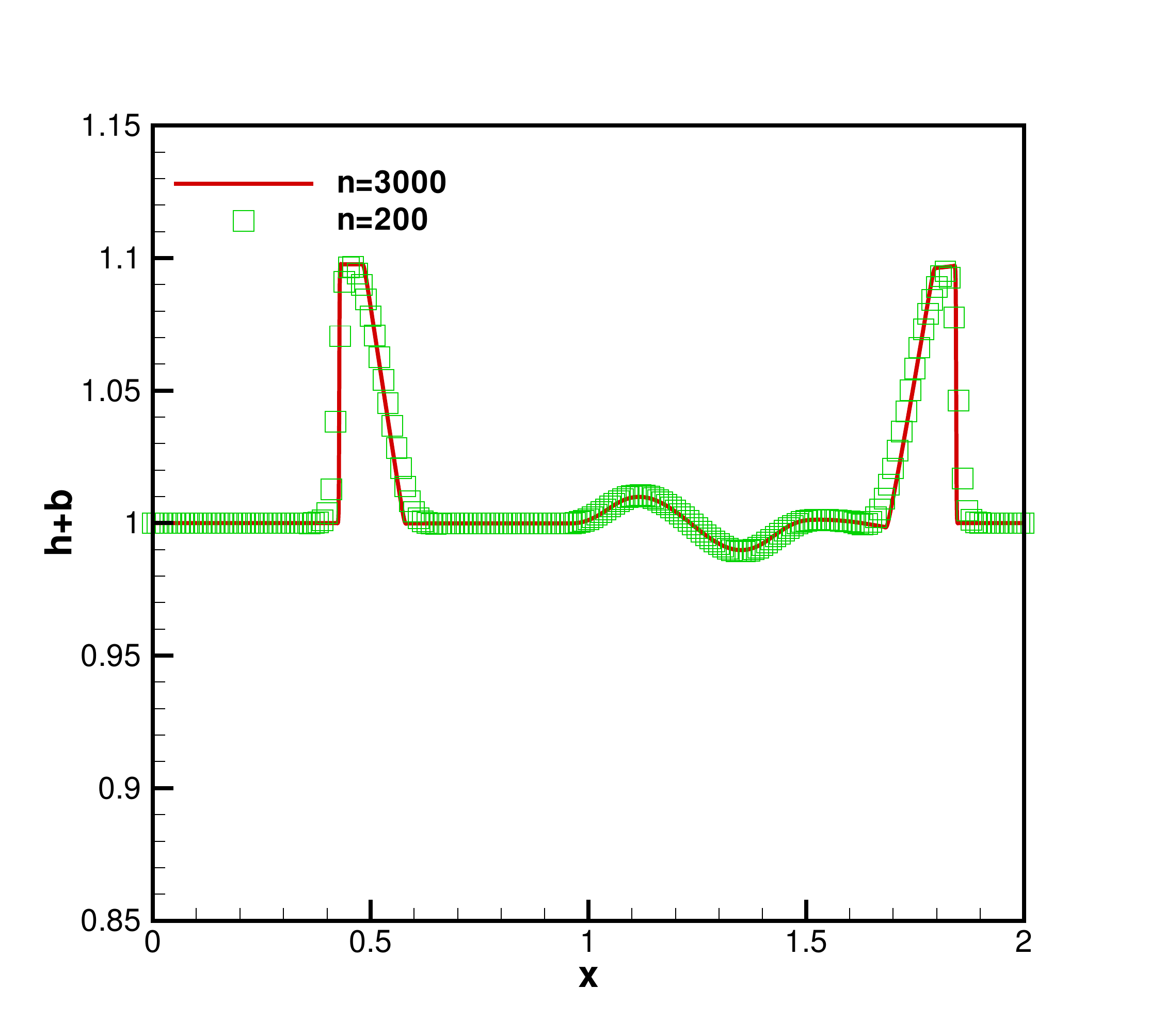}
\includegraphics[width=2.45in]{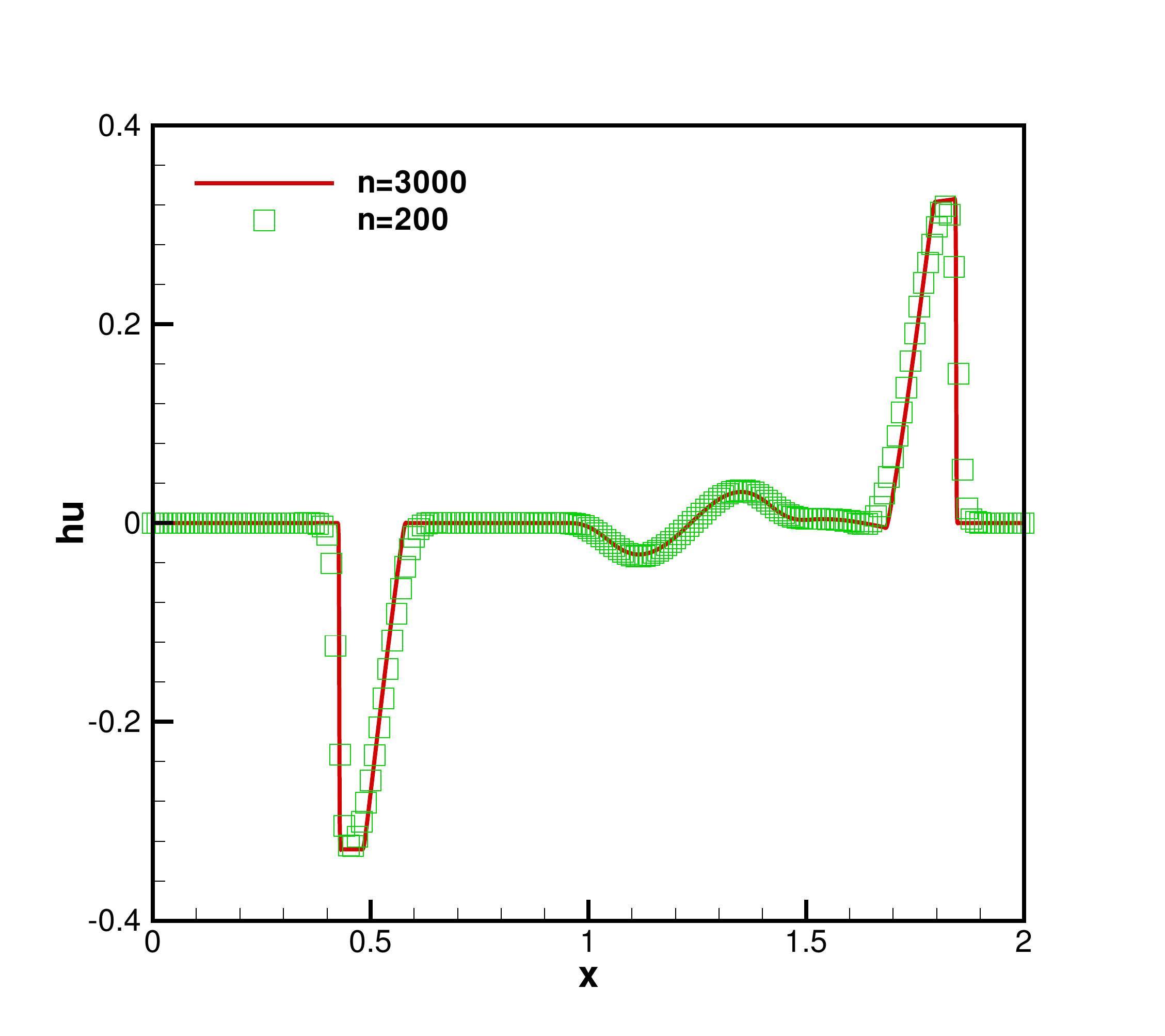}
\caption{Small perturbation of a steady state water flow with a big pulse, $t=0.2$ s.
Water surface level $h+b$ (left) and water discharge $hu$ (right).}\label{f:big-5}
\end{figure}


\begin{figure}
\centering
\includegraphics[width=2.45in]{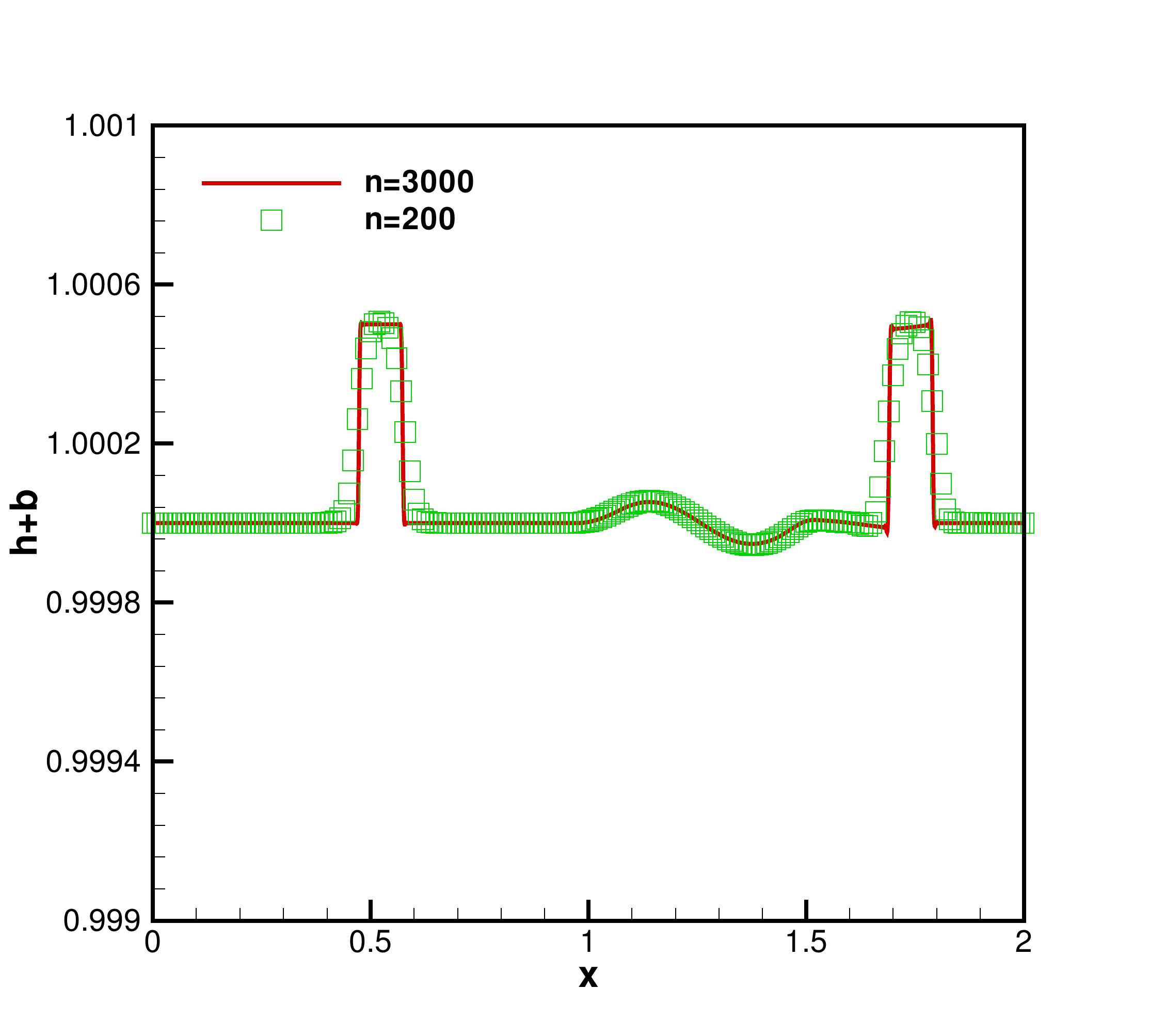}
\includegraphics[width=2.45in]{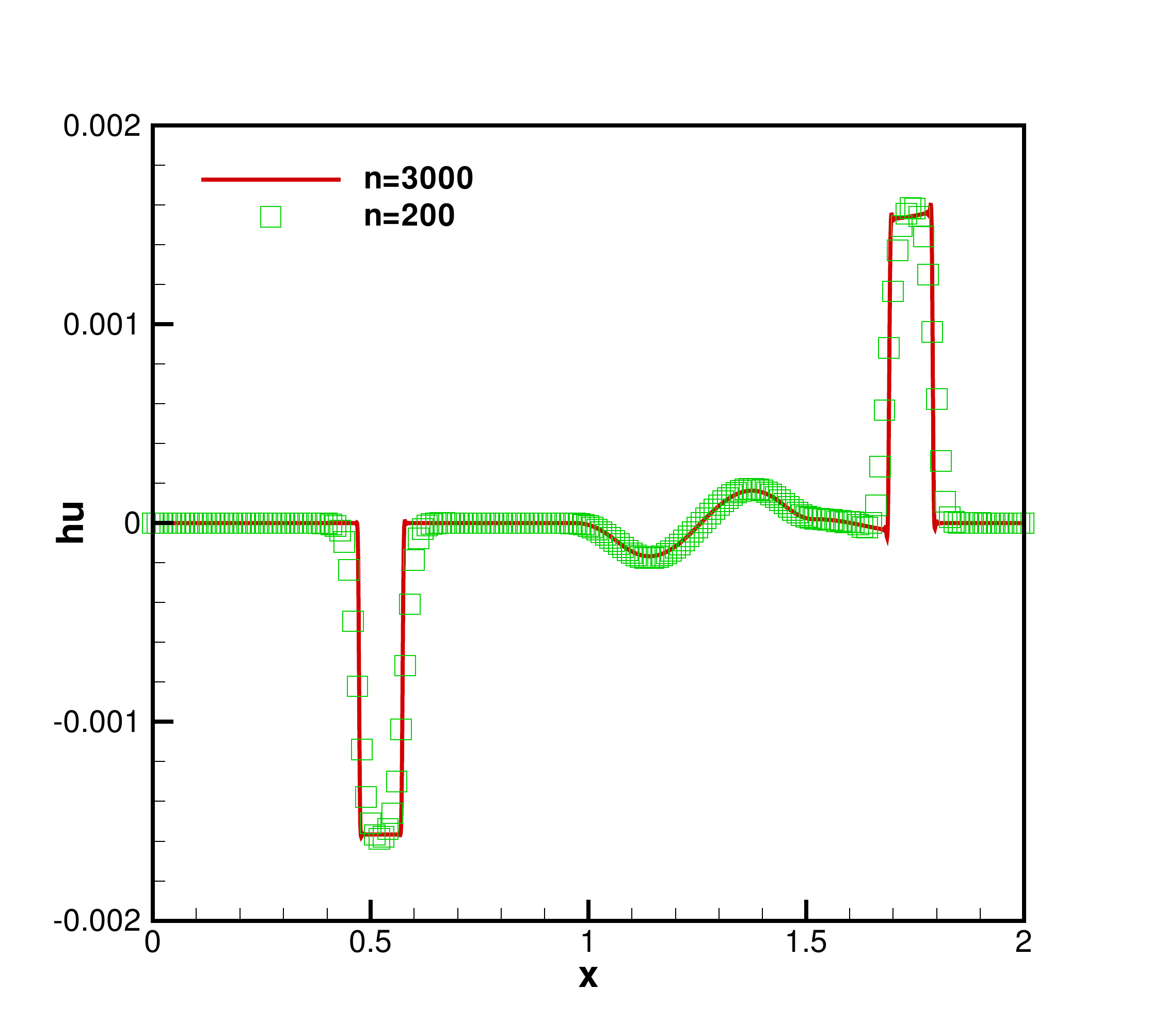}
\caption{Small perturbation of a steady state water flow with a small pulse, $t =0.2$ s. Water surface level $h+b$ (left) and water discharge $hu$ (right).}\label{f:small-5}
\end{figure}

\begin{figure}
\centering
\includegraphics[width=2.45in]{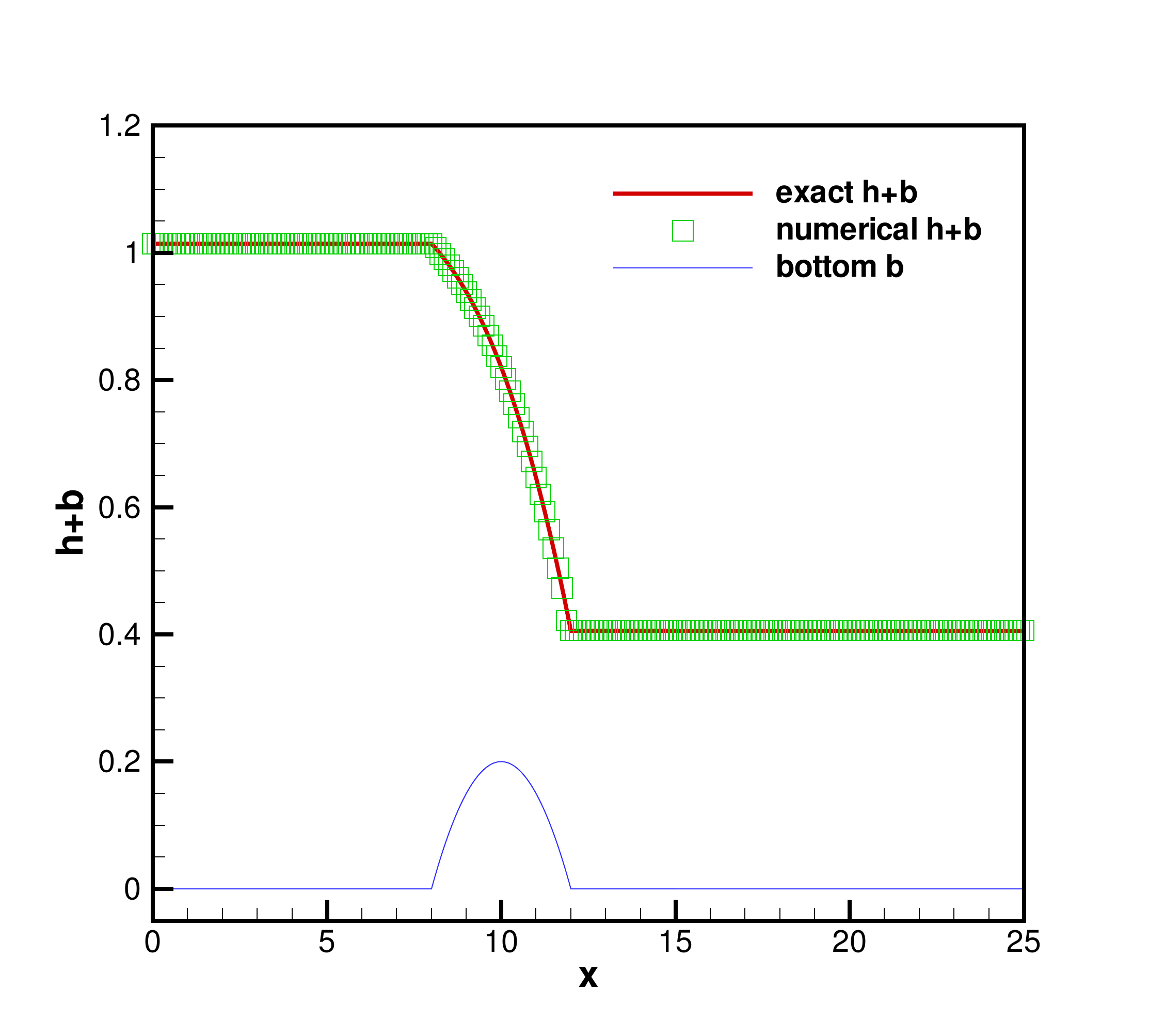}
\includegraphics[width=2.45in]{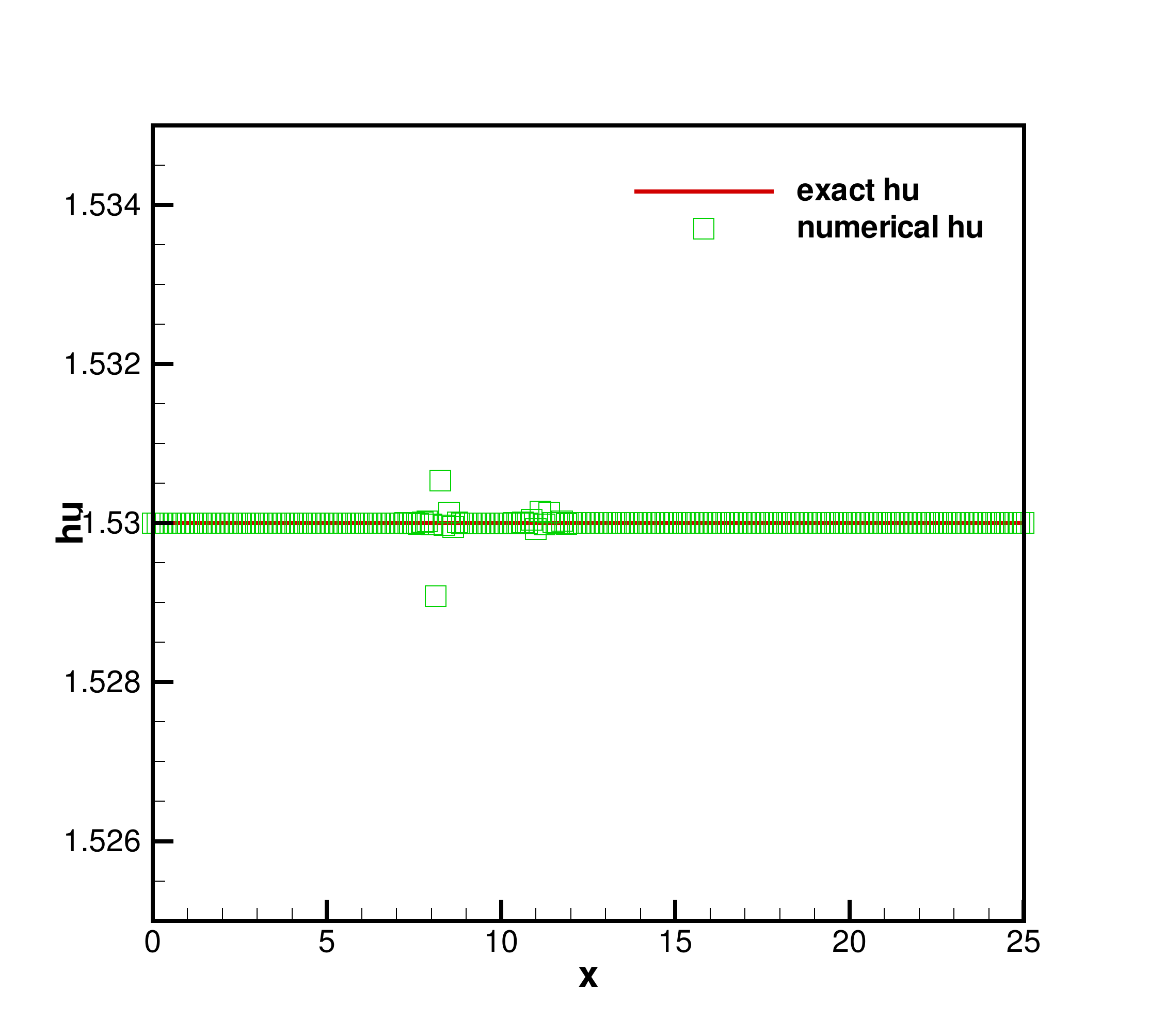}
\caption{Transcritical flow without a shock, $t=200$ s. Water
surface level $h+b$ (left) and water discharge $hu$ (right).}  \label{f:bump-a}
\end{figure}

\begin{figure}
\centering
\includegraphics[width=2.45in]{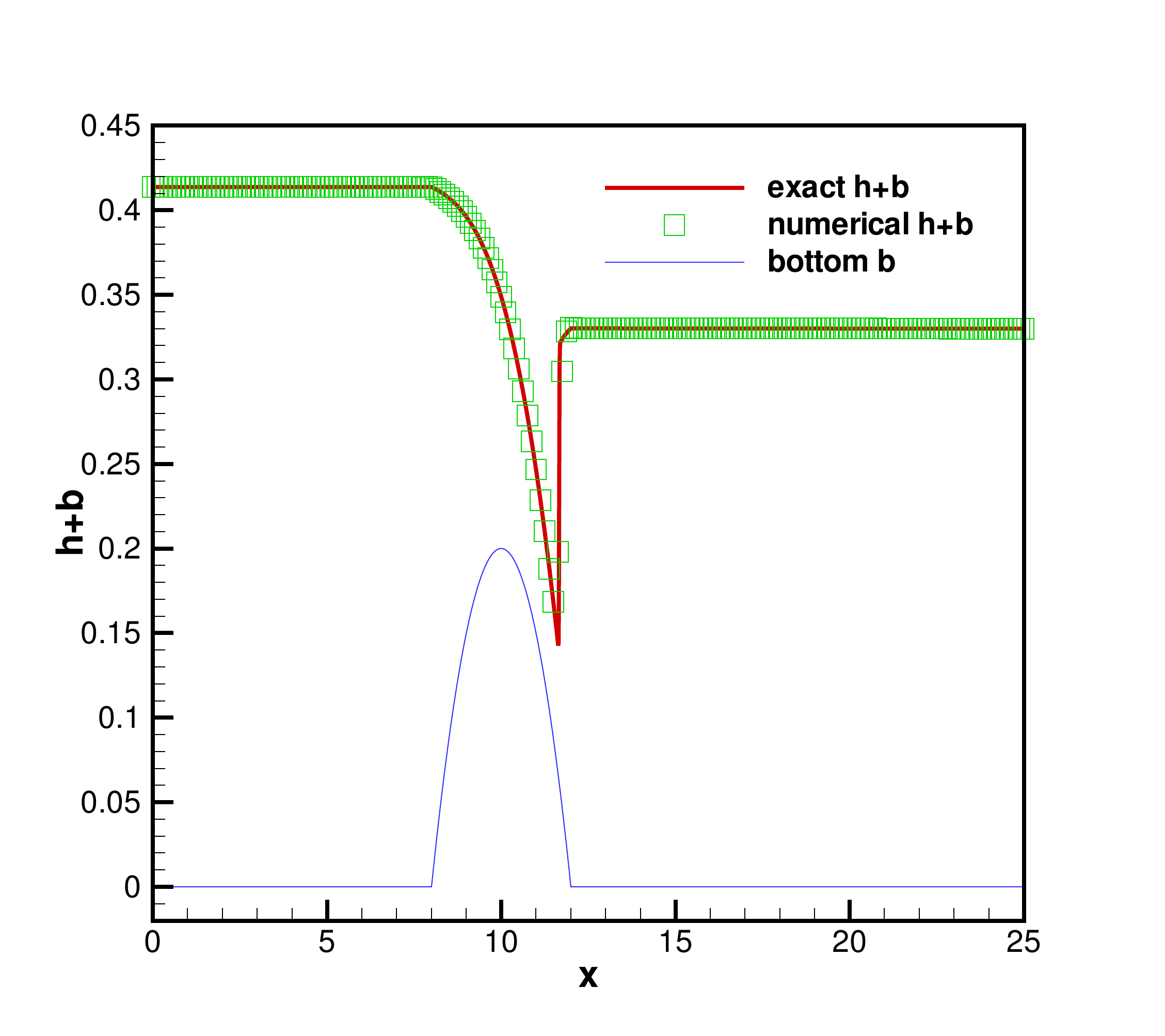}
\includegraphics[width=2.45in]{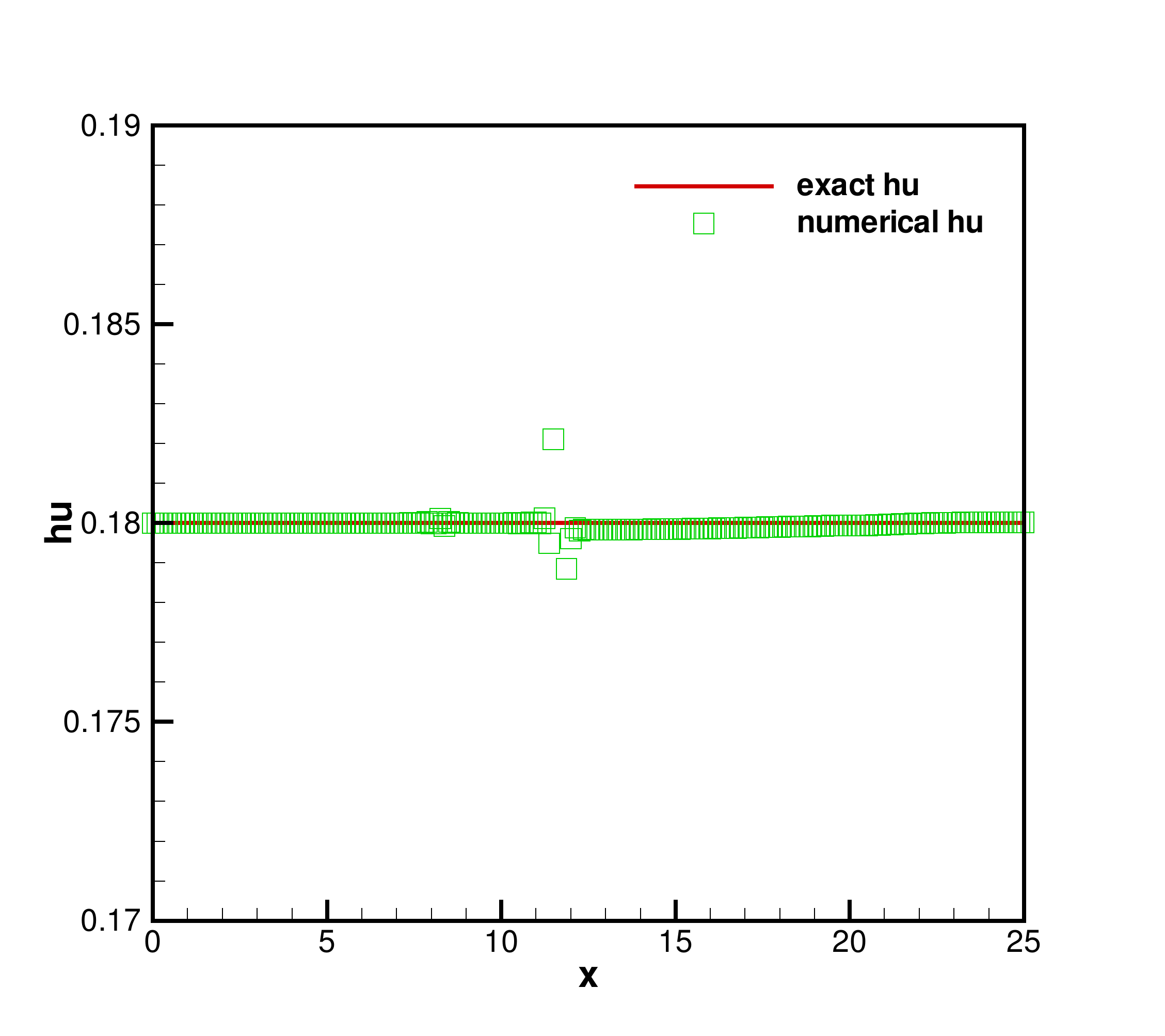}
\caption{Transcritical flow with a shock, $t=200$ s. Water
surface level $h+b$ (left) and water discharge $hu$ (right).}  \label{f:bump-b}
\end{figure}

\begin{figure}
\centering
\includegraphics[width=2.45in]{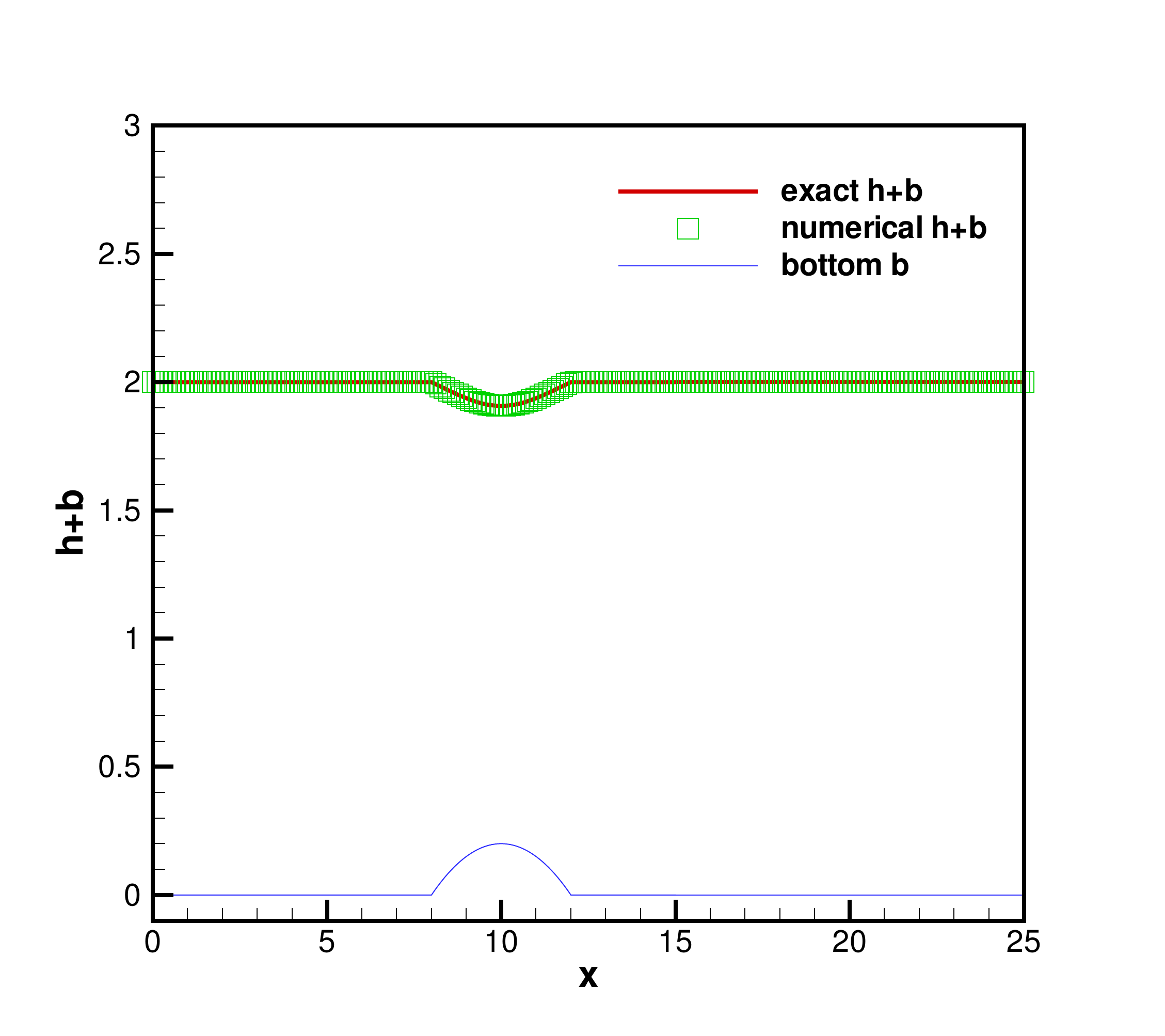}
\includegraphics[width=2.45in]{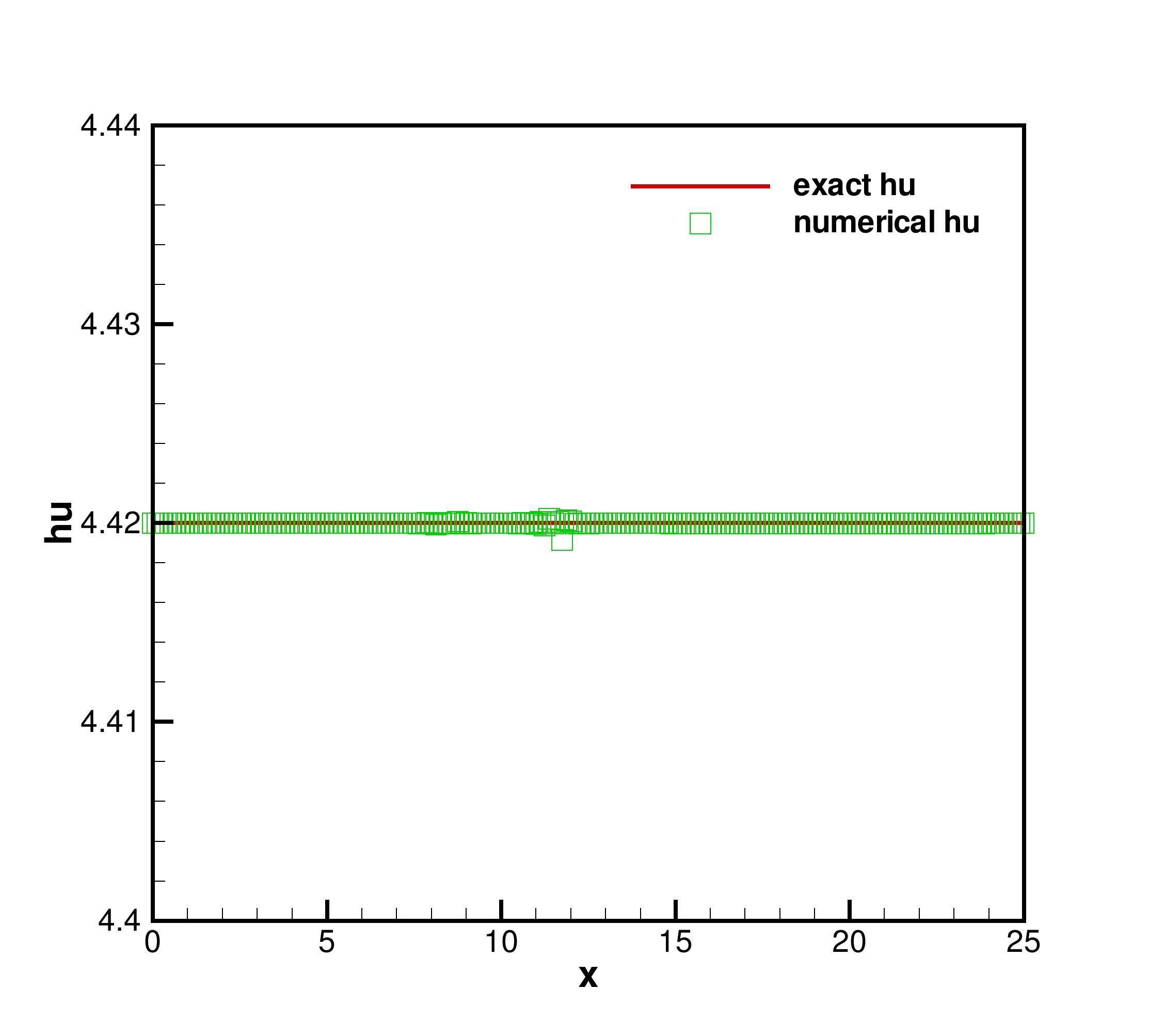}
\caption{Subcritical flow, $t=200$ s. Water
surface level $h+b$ (left) and water discharge $hu$ (right).} \label{f:bump-c}
\end{figure}

\begin{figure}
\centering
\includegraphics[width=2.45in]{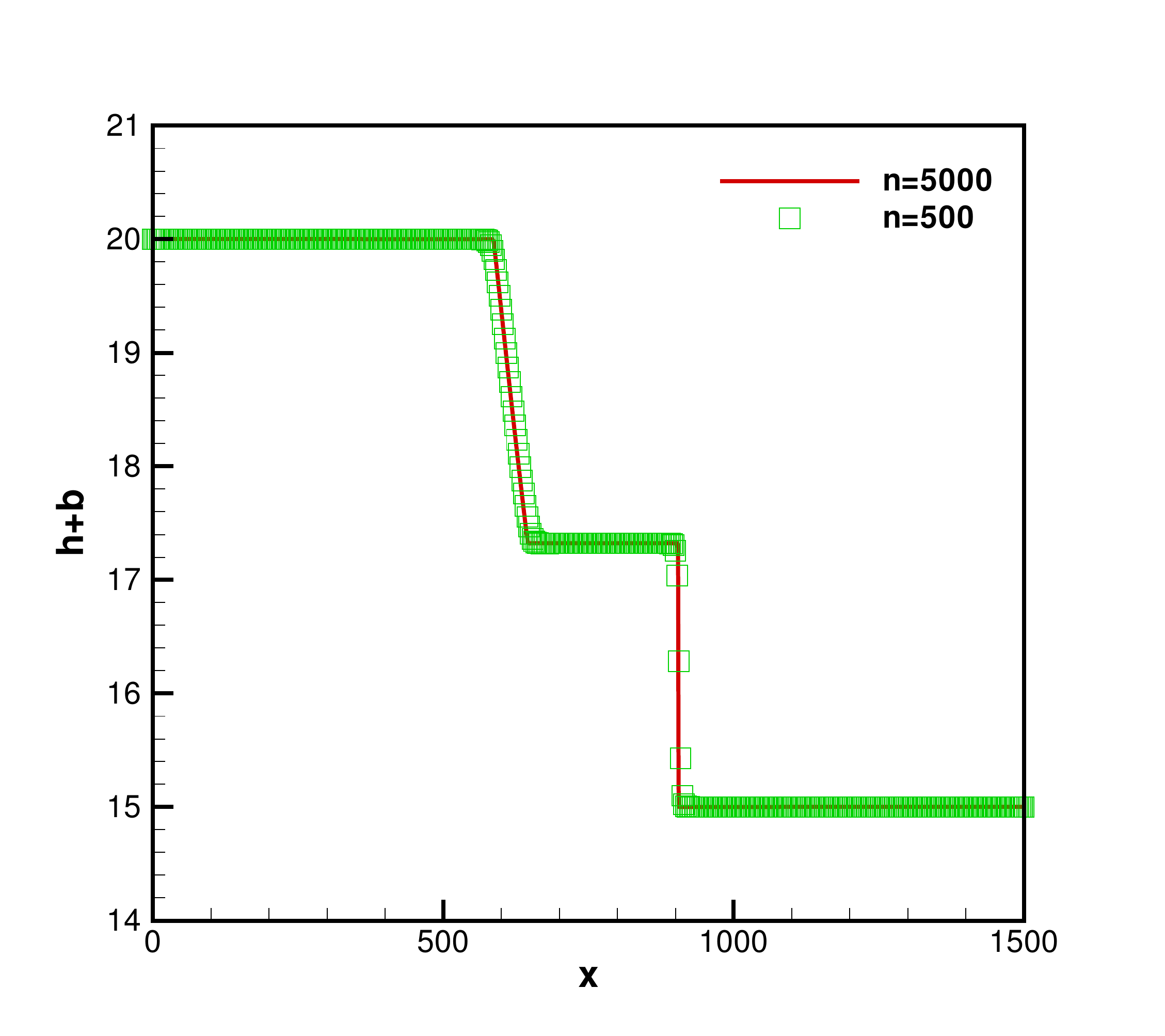}
\includegraphics[width=2.45in]{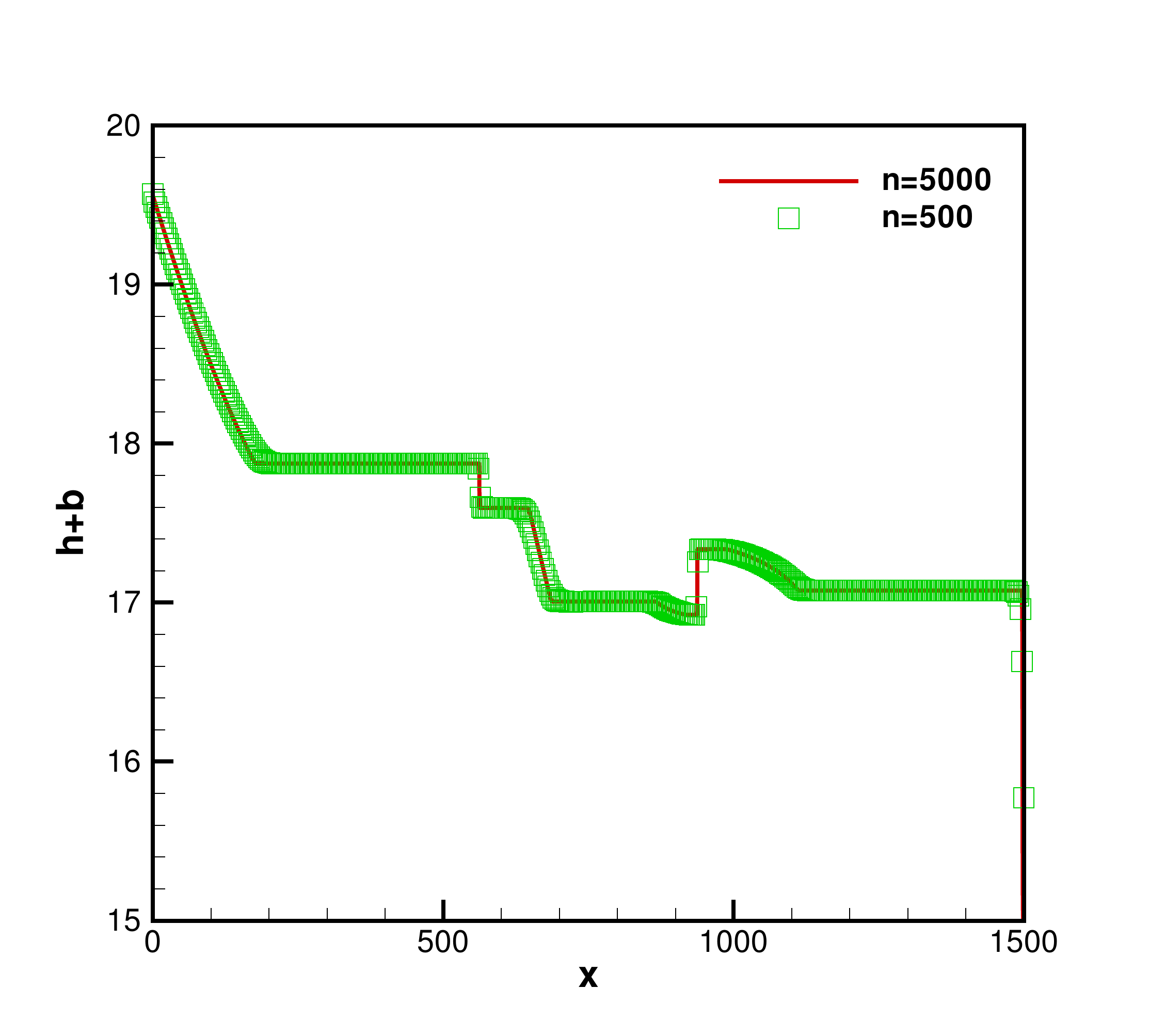}
\caption{The dam break problem over a rectangular bump. Water
surface level $h+b$ at $t=15$ s (left) and $t=60$ s (right).}\label{rectangular-bump}
\end{figure}

\begin{figure}
\centering
\includegraphics[width=2.45in]{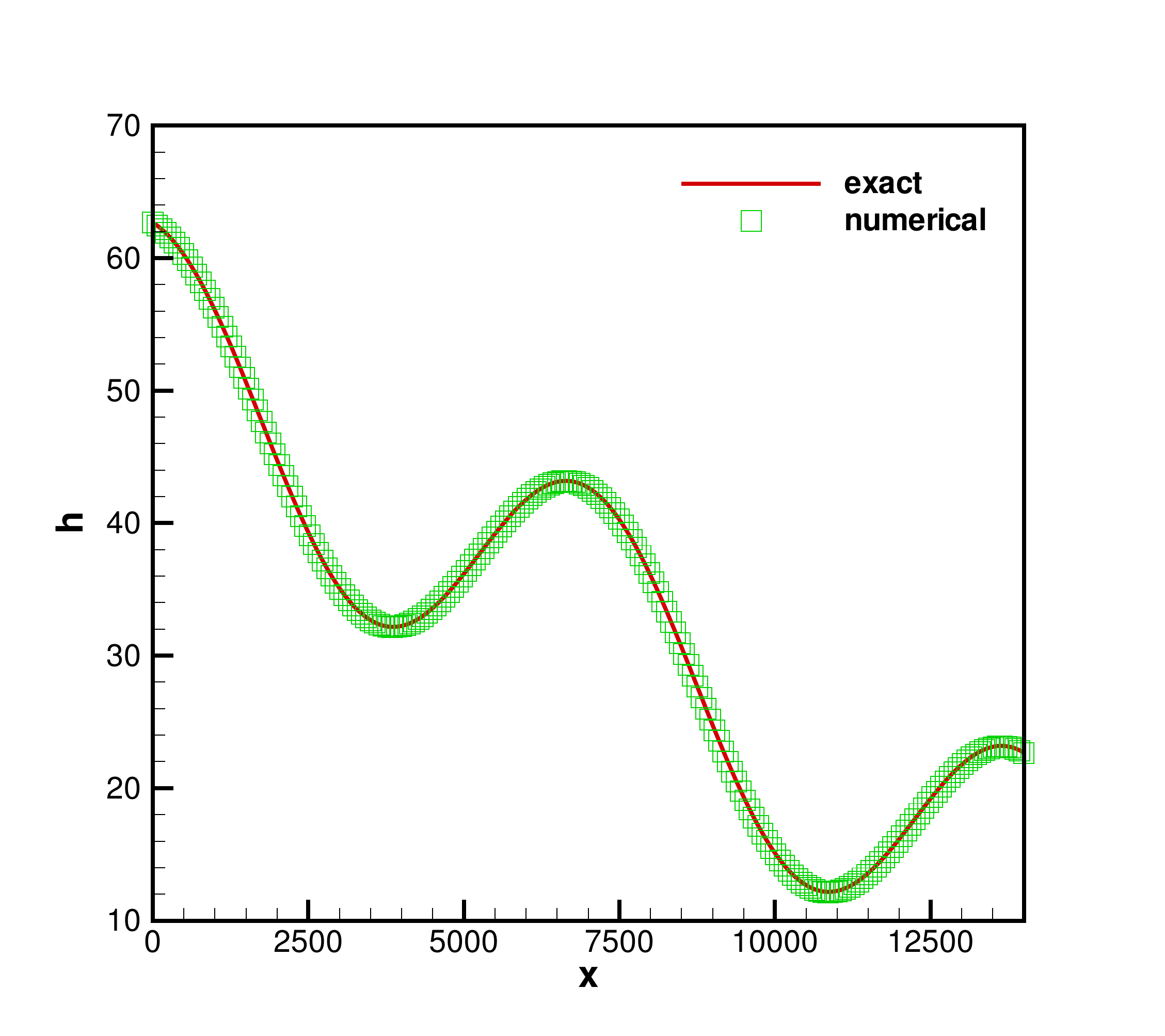}
\includegraphics[width=2.45in]{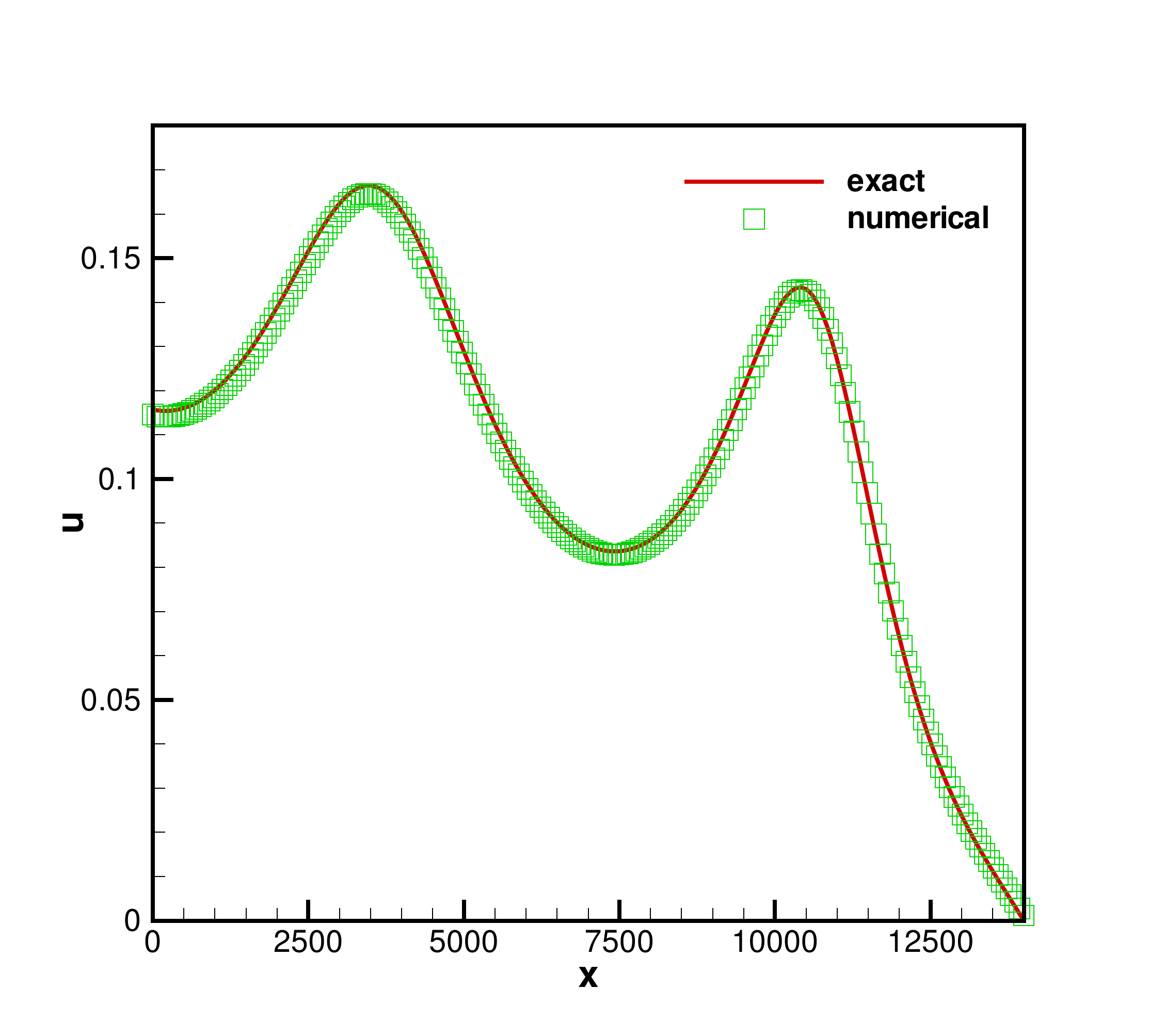}
\caption{The tidal wave flow, $t = 7552.13$ s. Water
depth $h$ (left) and water velocity $u$ (right).} \label{tidal-flow}
\end{figure}

\begin{figure}
\centering
\includegraphics[width=2.45in]{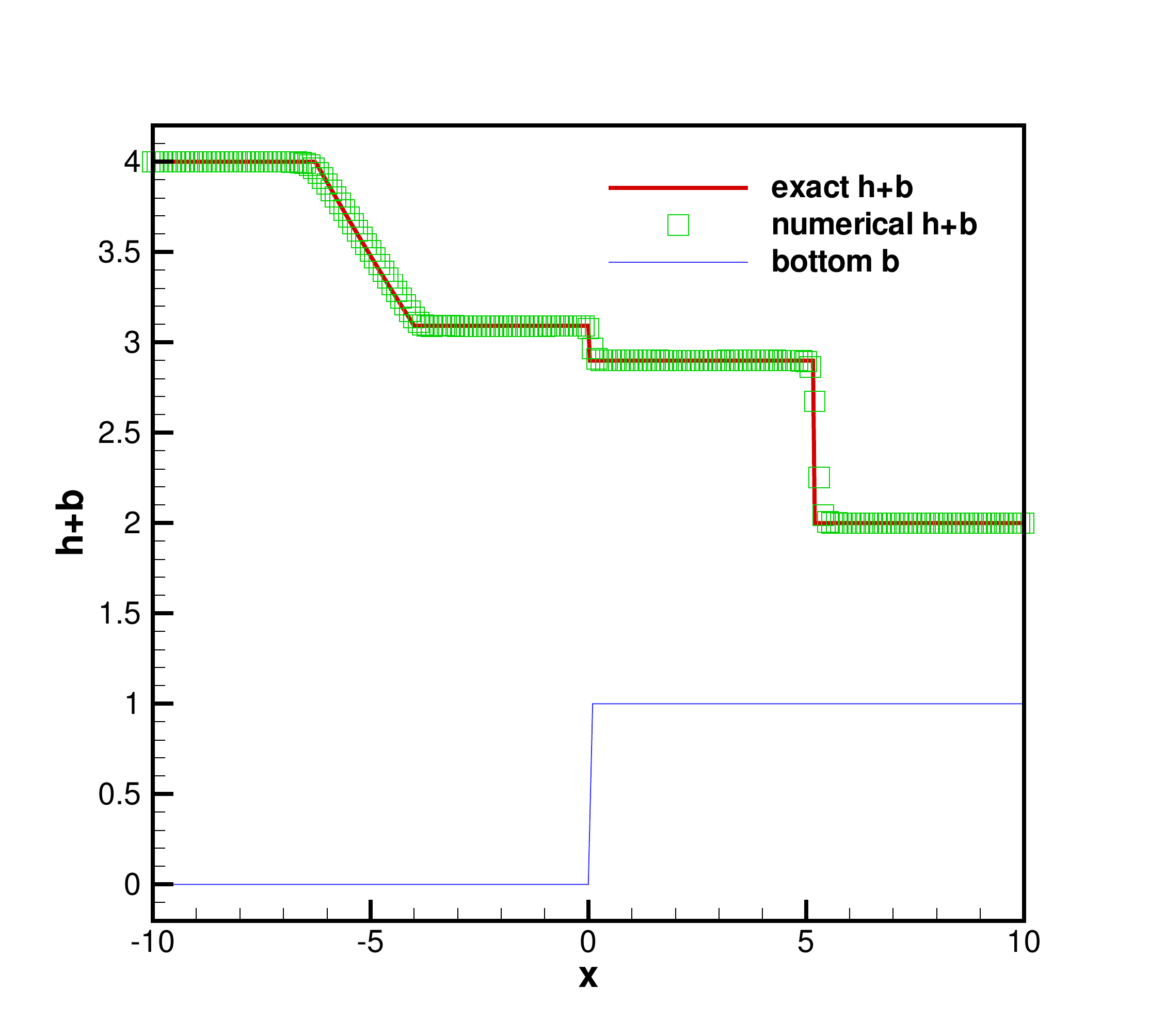}
\includegraphics[width=2.45in]{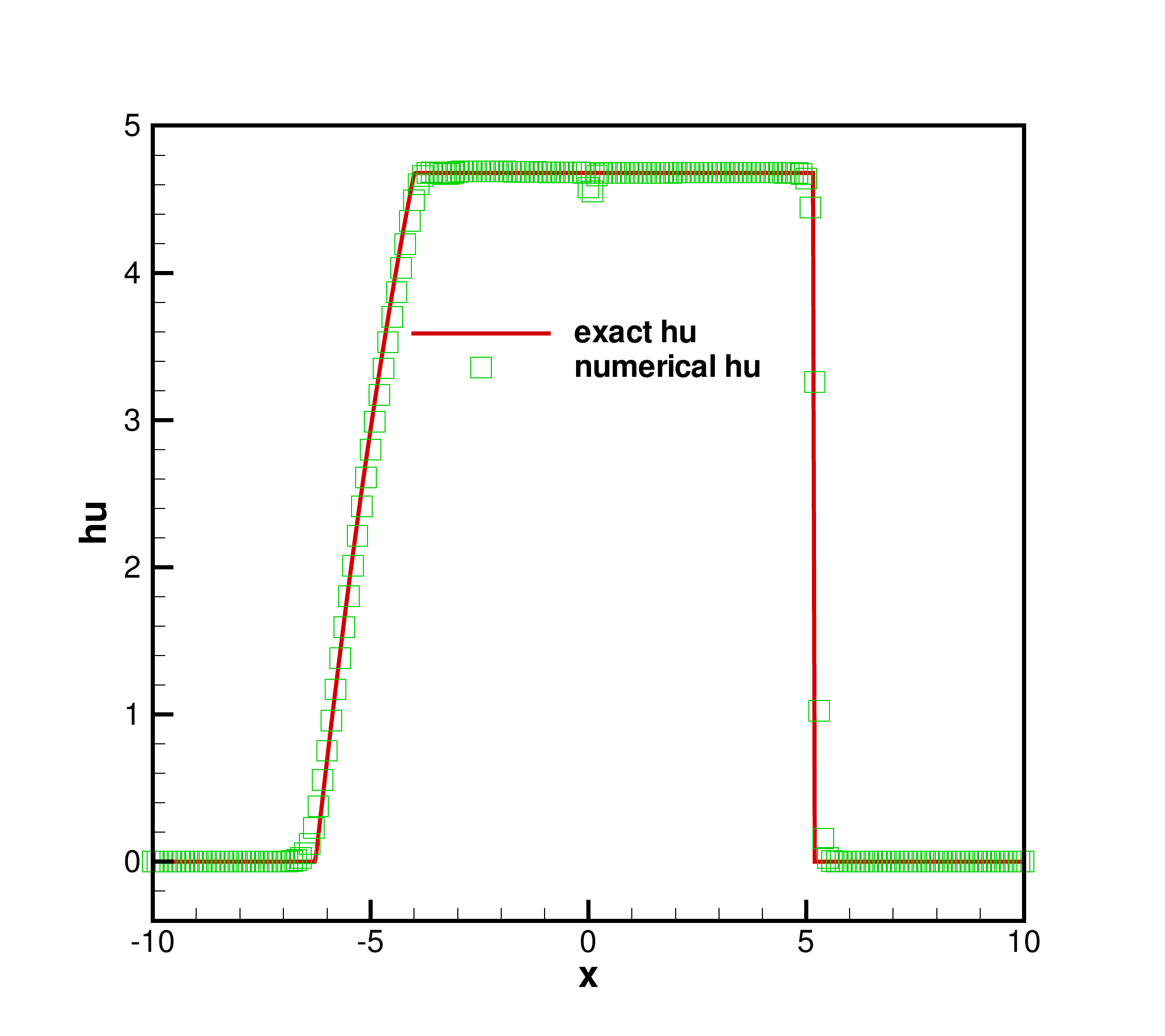}
\caption{$1$-rarefaction and $2$-shock problem, $t=1$ s. Water
surface level $h+b$ (left) and water discharge $hu$ (right).} \label{f:1D-step-RS}
\end{figure}

\begin{figure}
\centering
\includegraphics[width=2.45in]{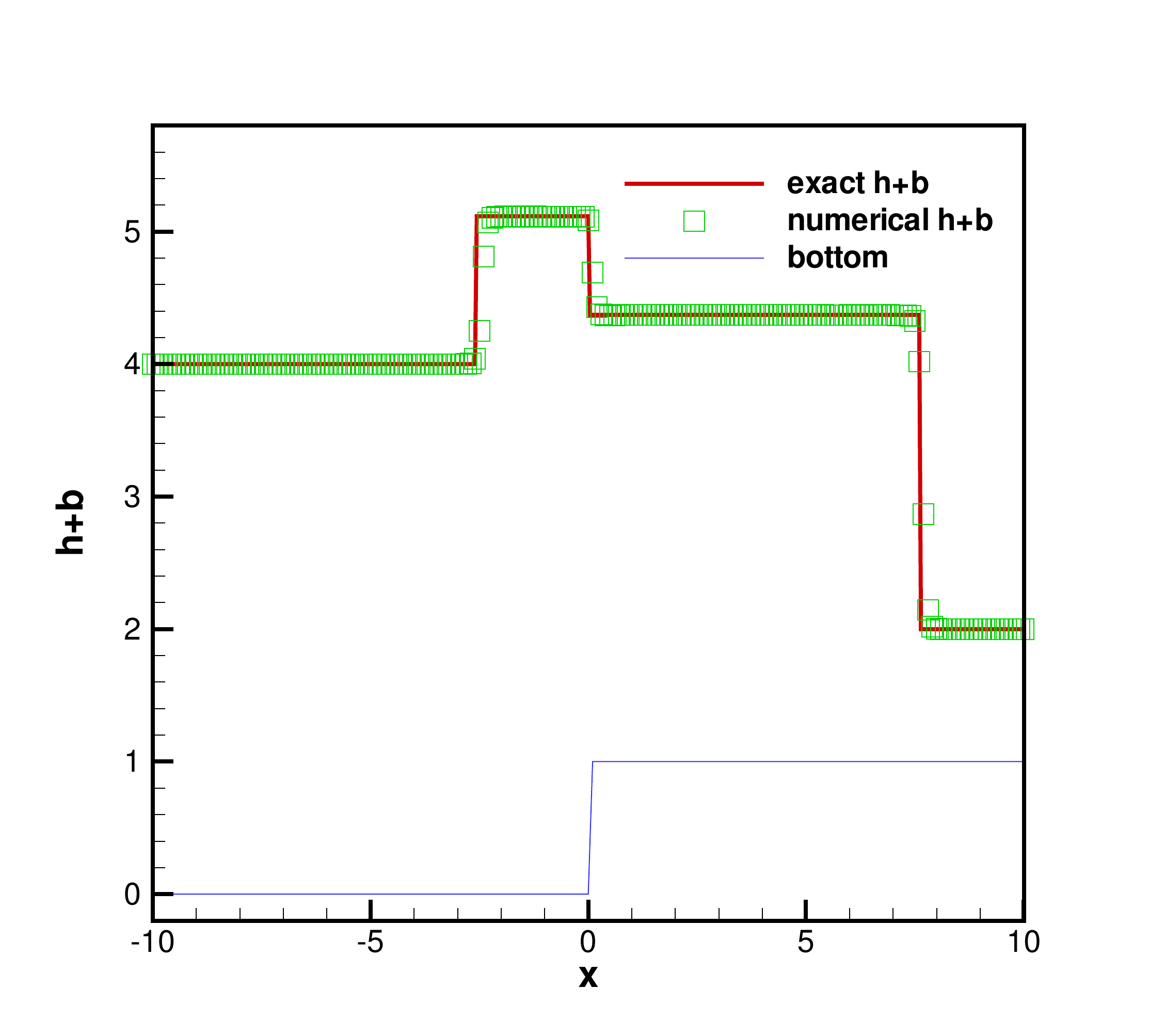}
\includegraphics[width=2.45in]{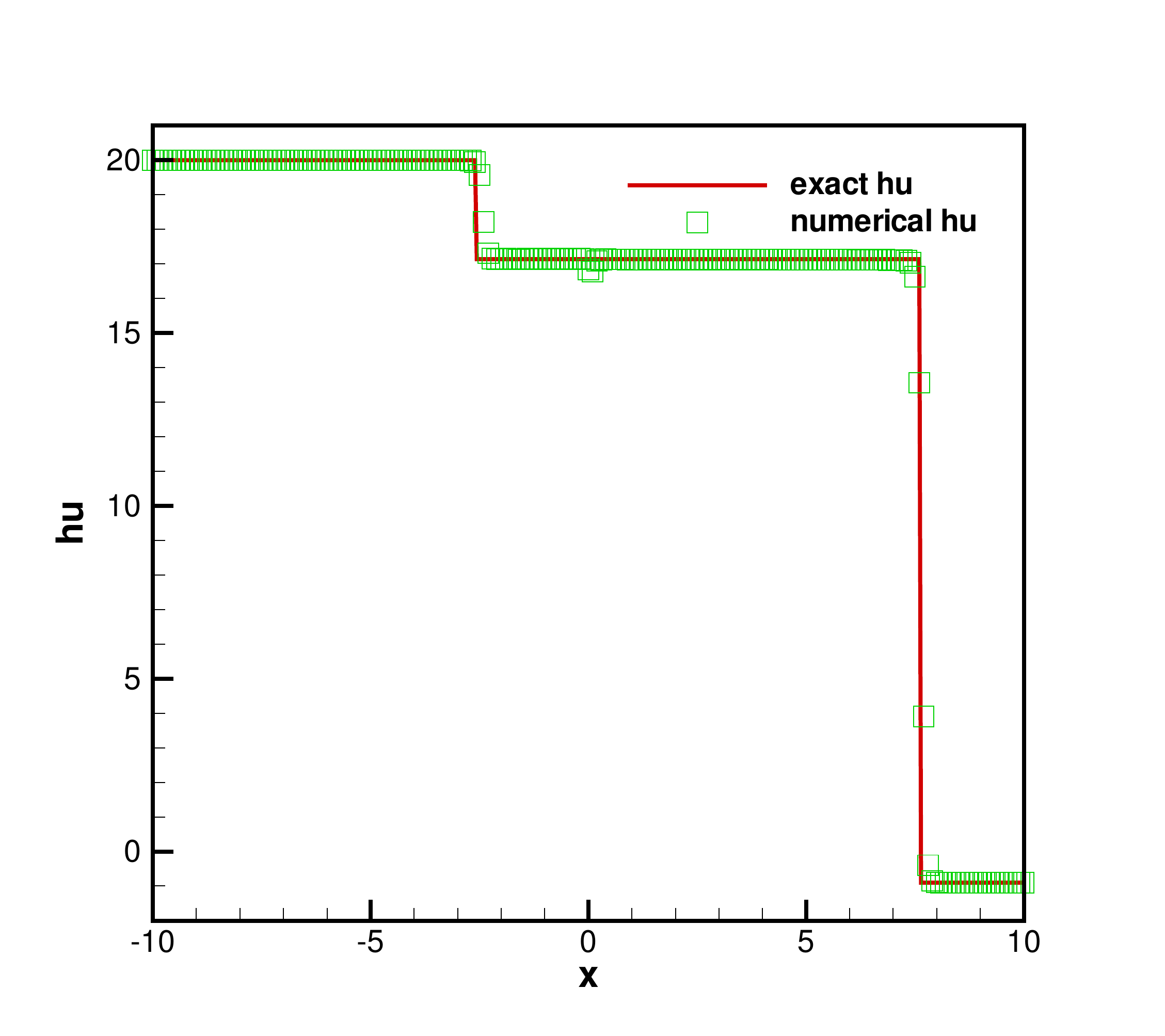}
\caption{1-shock and 2-shock problem, $t=1$ s. Water
surface level $h+b$ (left) and water discharge $hu$ (right).}
\label{f:1D-step-SS}
\end{figure}

\begin{figure}
\includegraphics[width=3.0in]{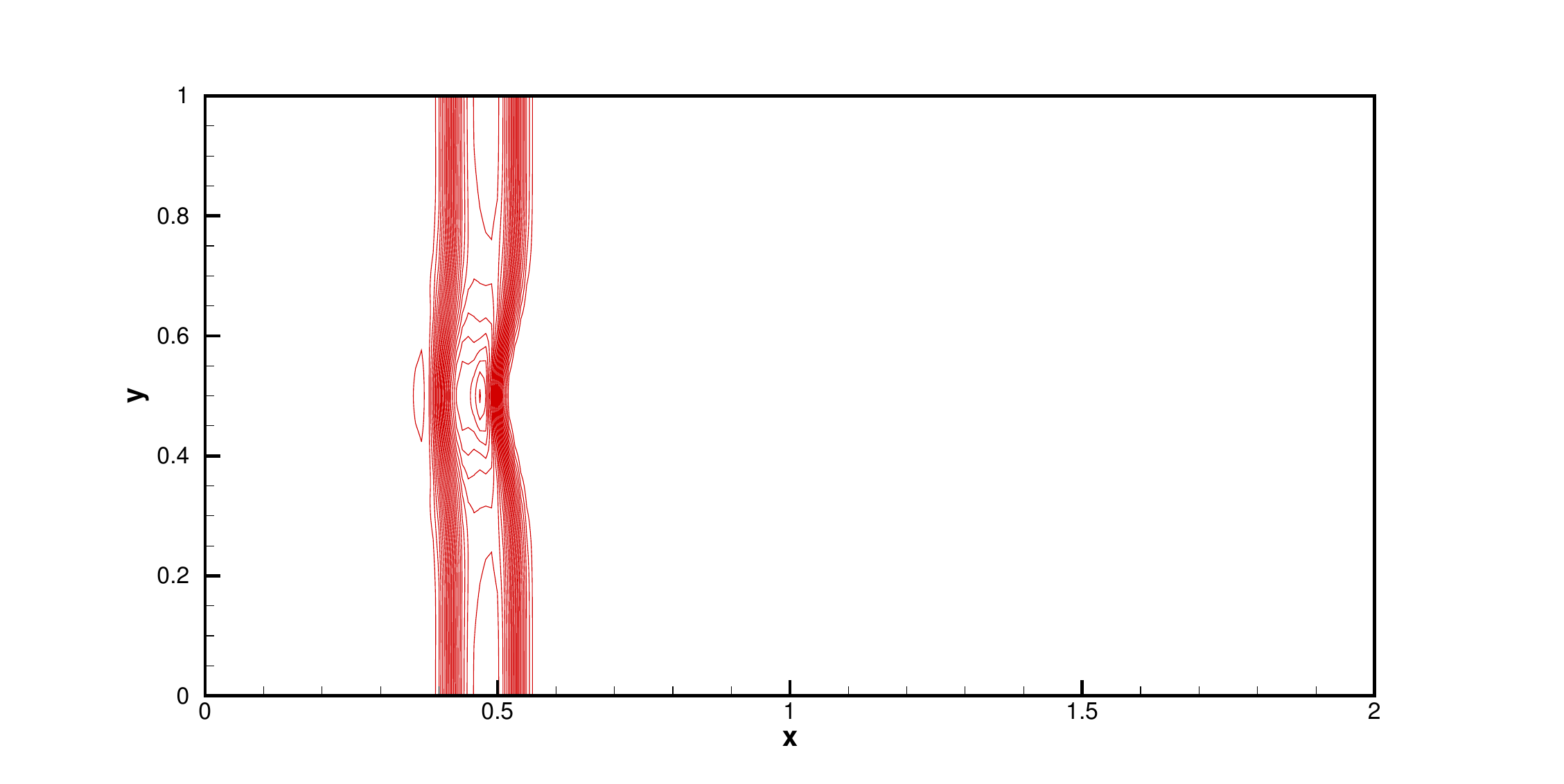}
\includegraphics[width=3.0in]{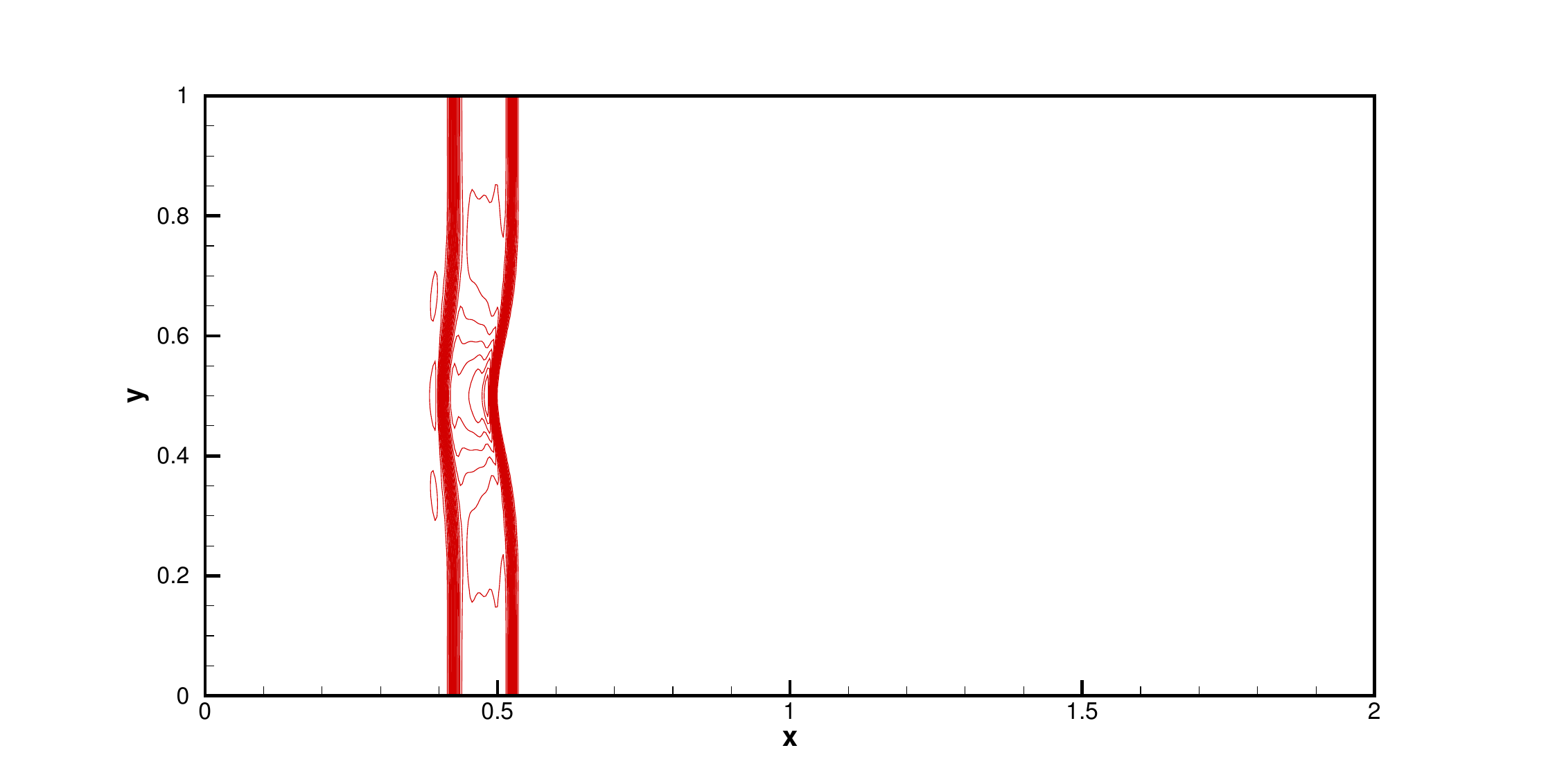}
\includegraphics[width=3.0in]{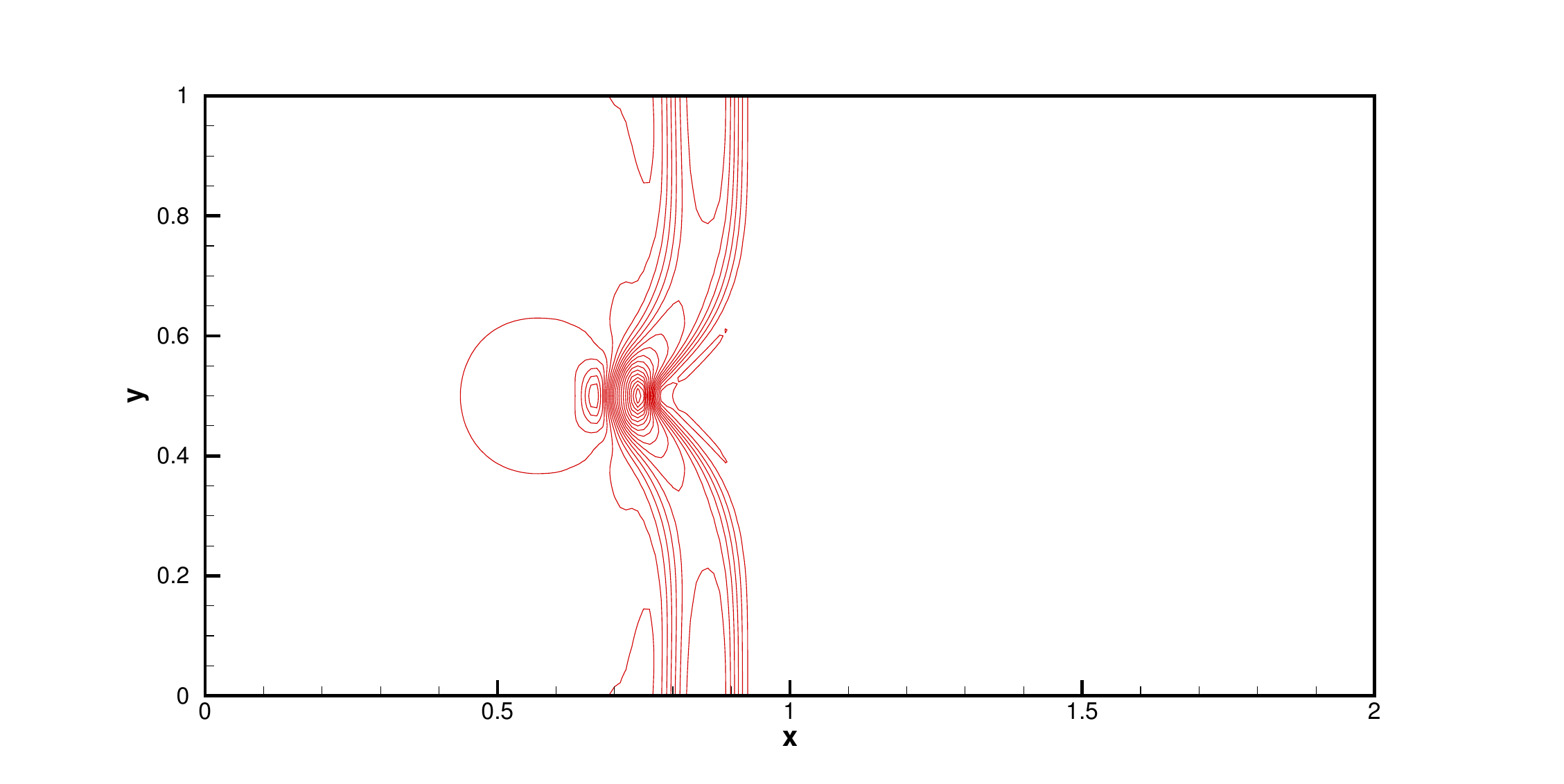}
\includegraphics[width=3.0in]{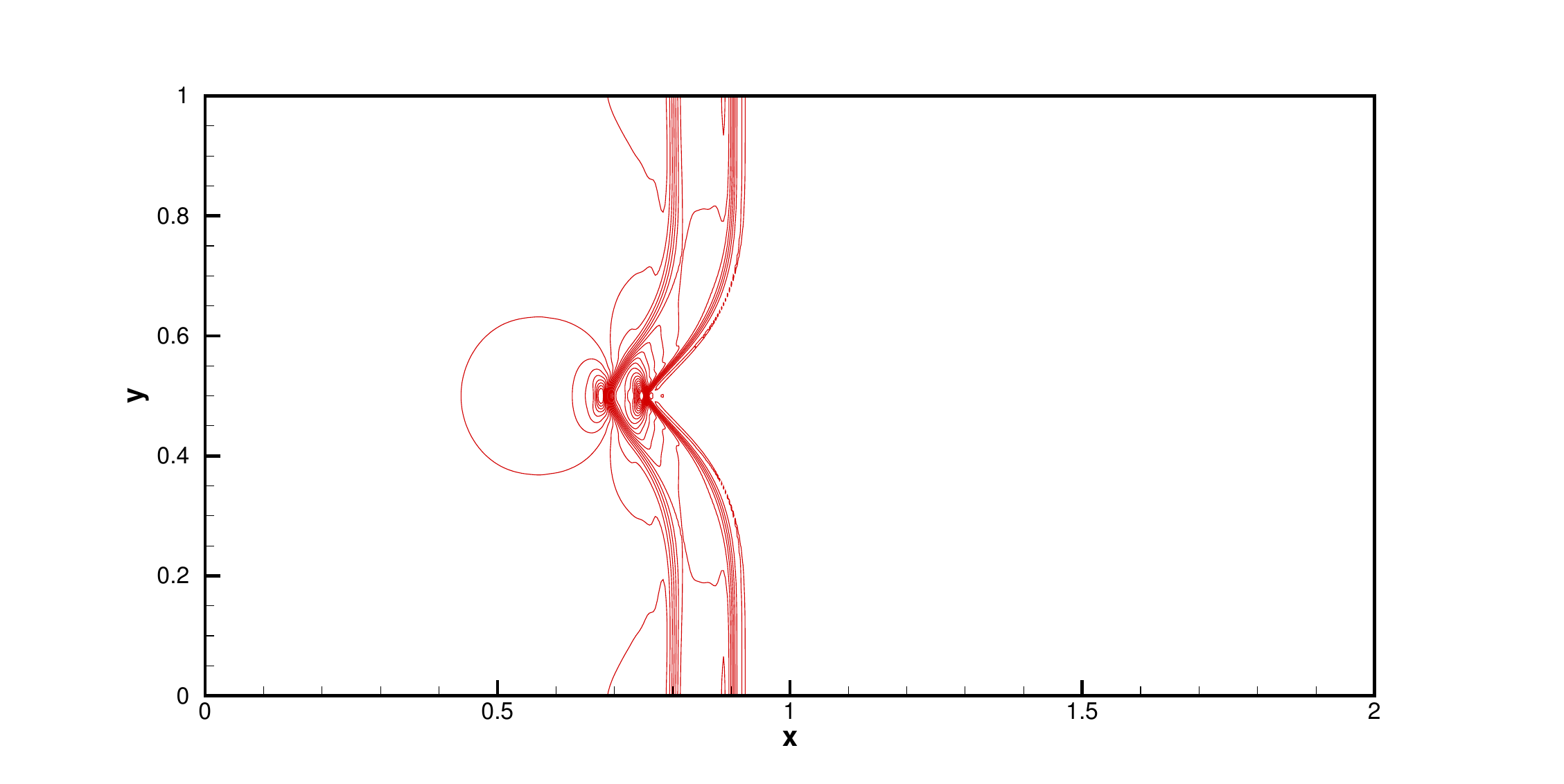}
\includegraphics[width=3.0in]{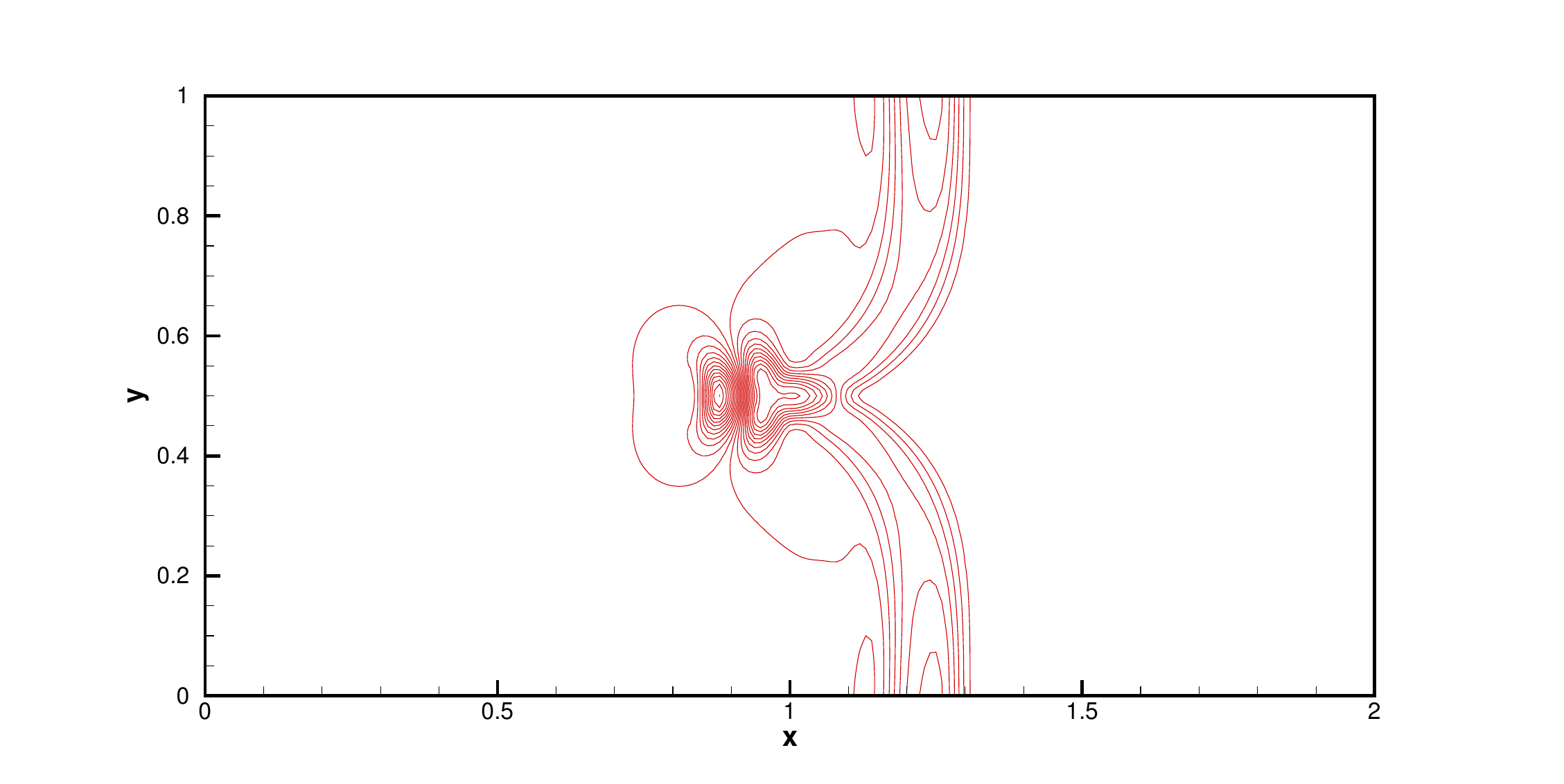}
\includegraphics[width=3.0in]{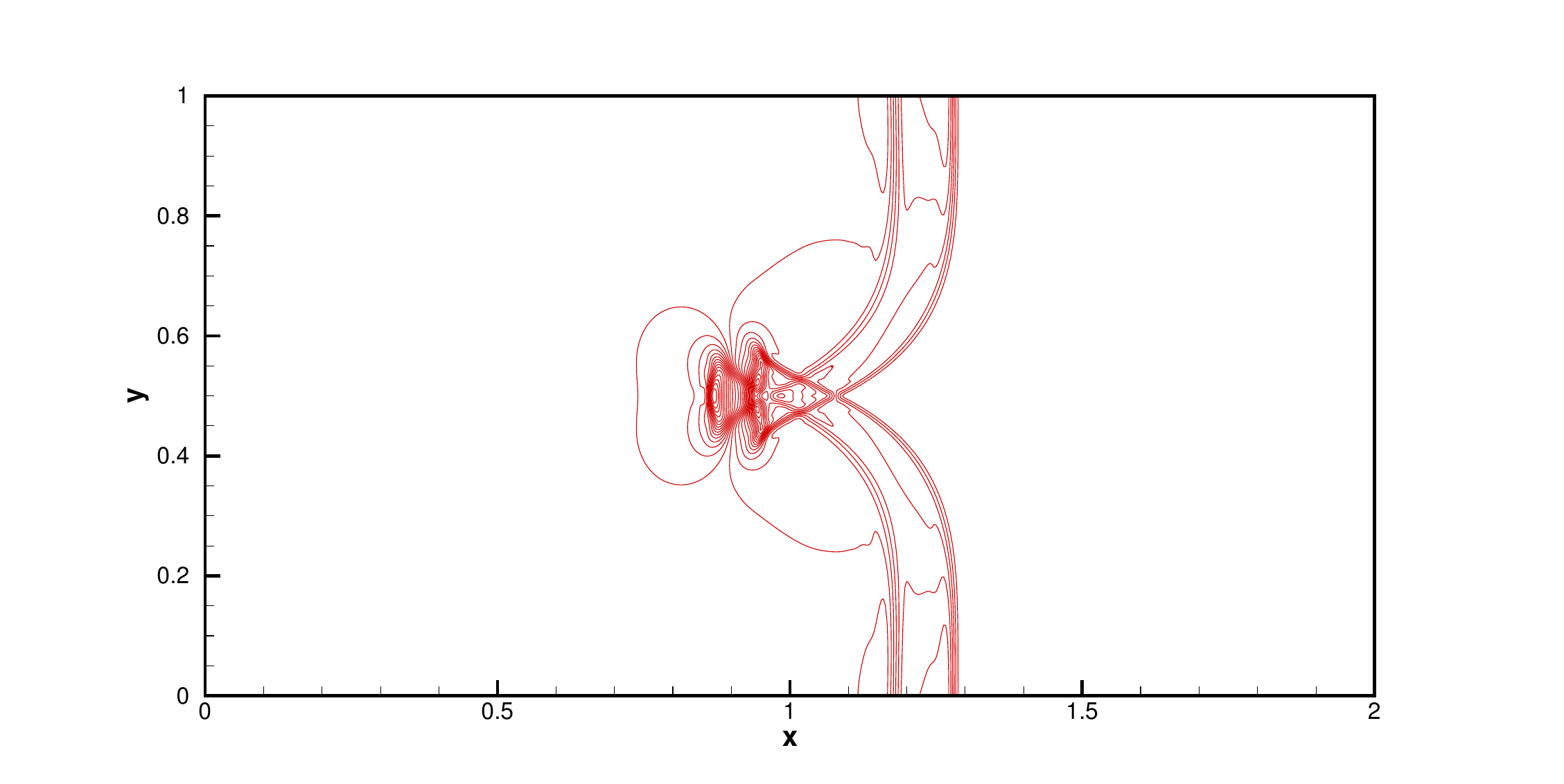}
\includegraphics[width=3.0in]{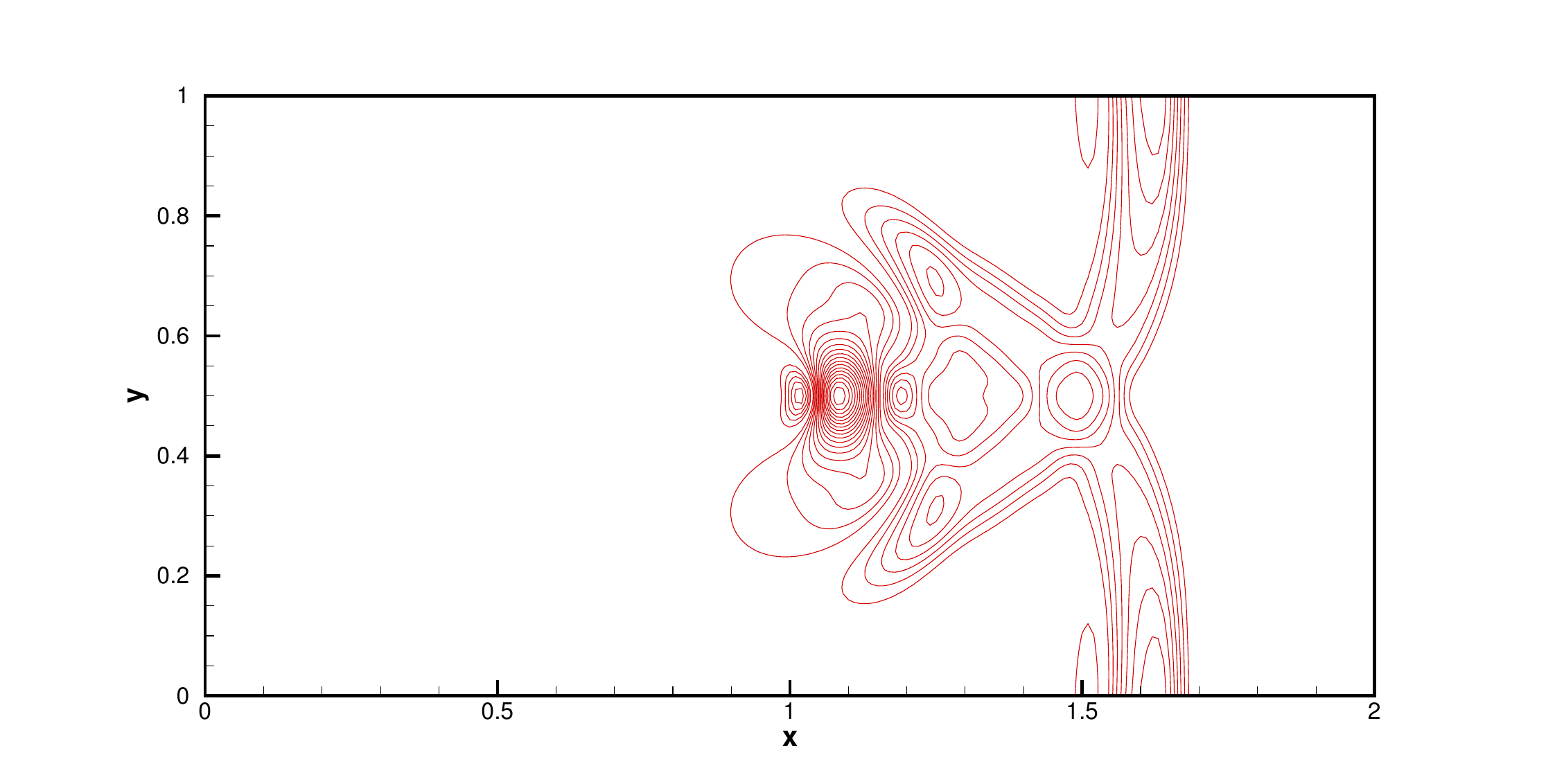}
\includegraphics[width=3.0in]{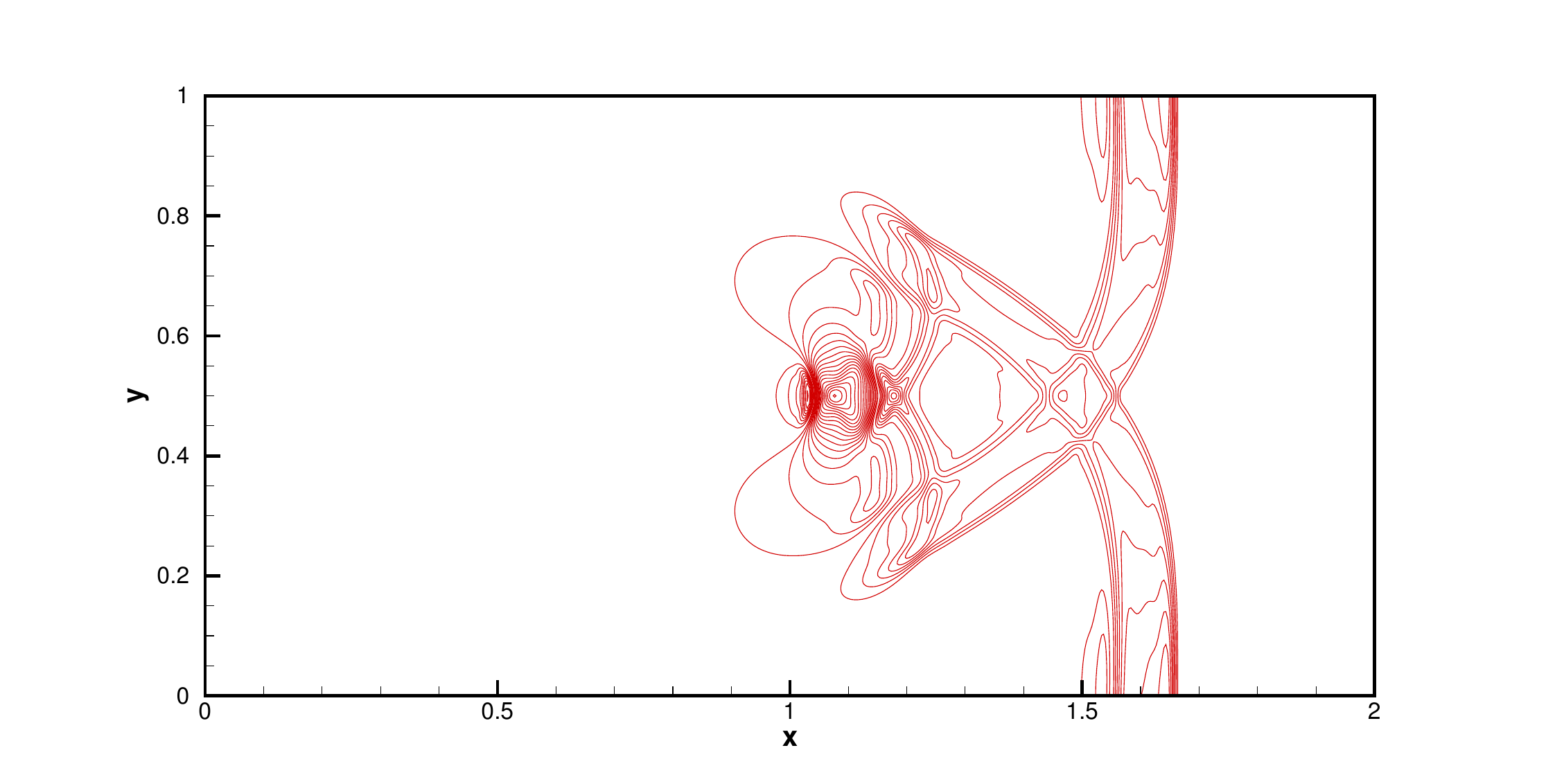}
\includegraphics[width=3.0in]{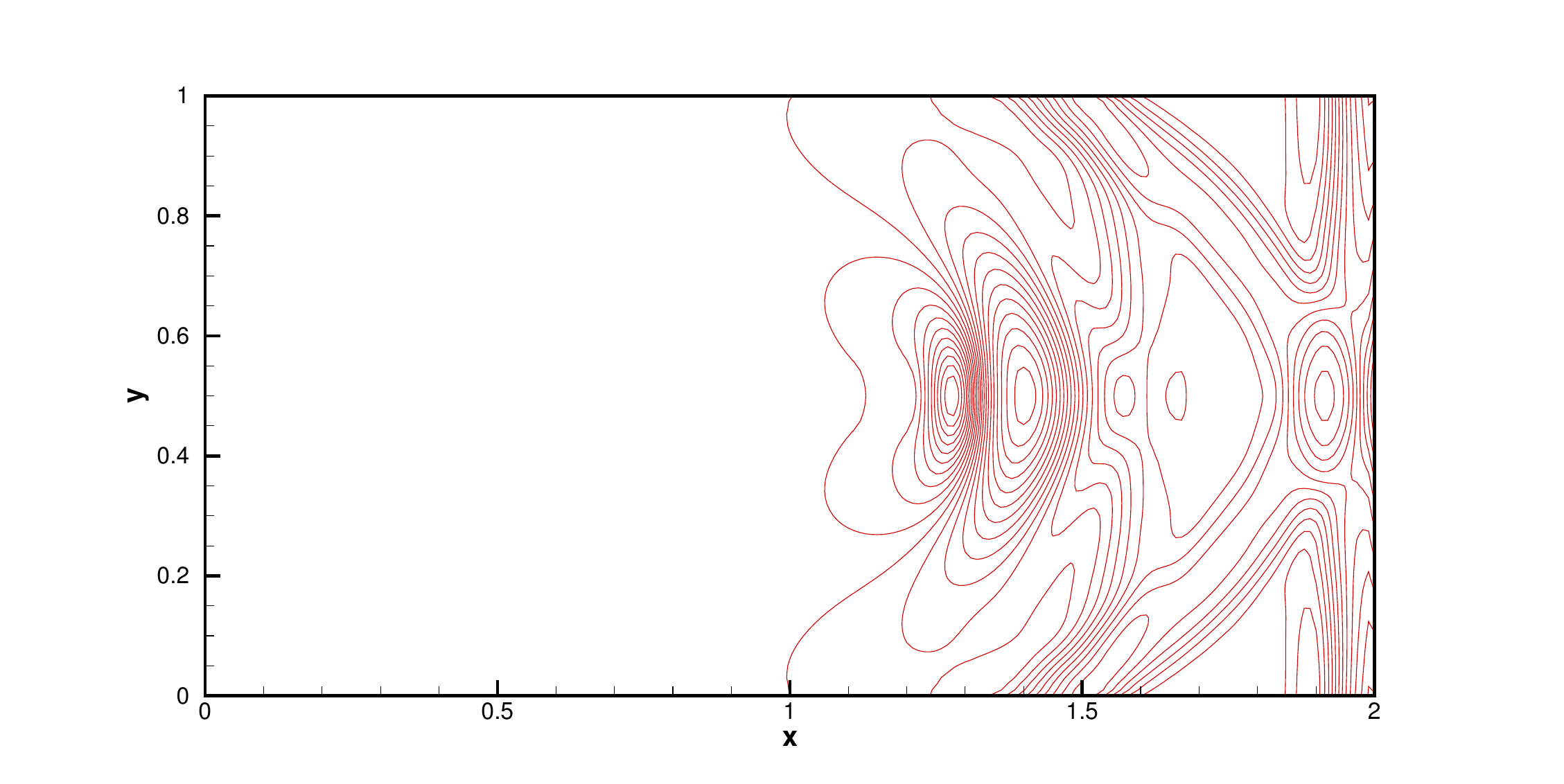}
\includegraphics[width=3.0in]{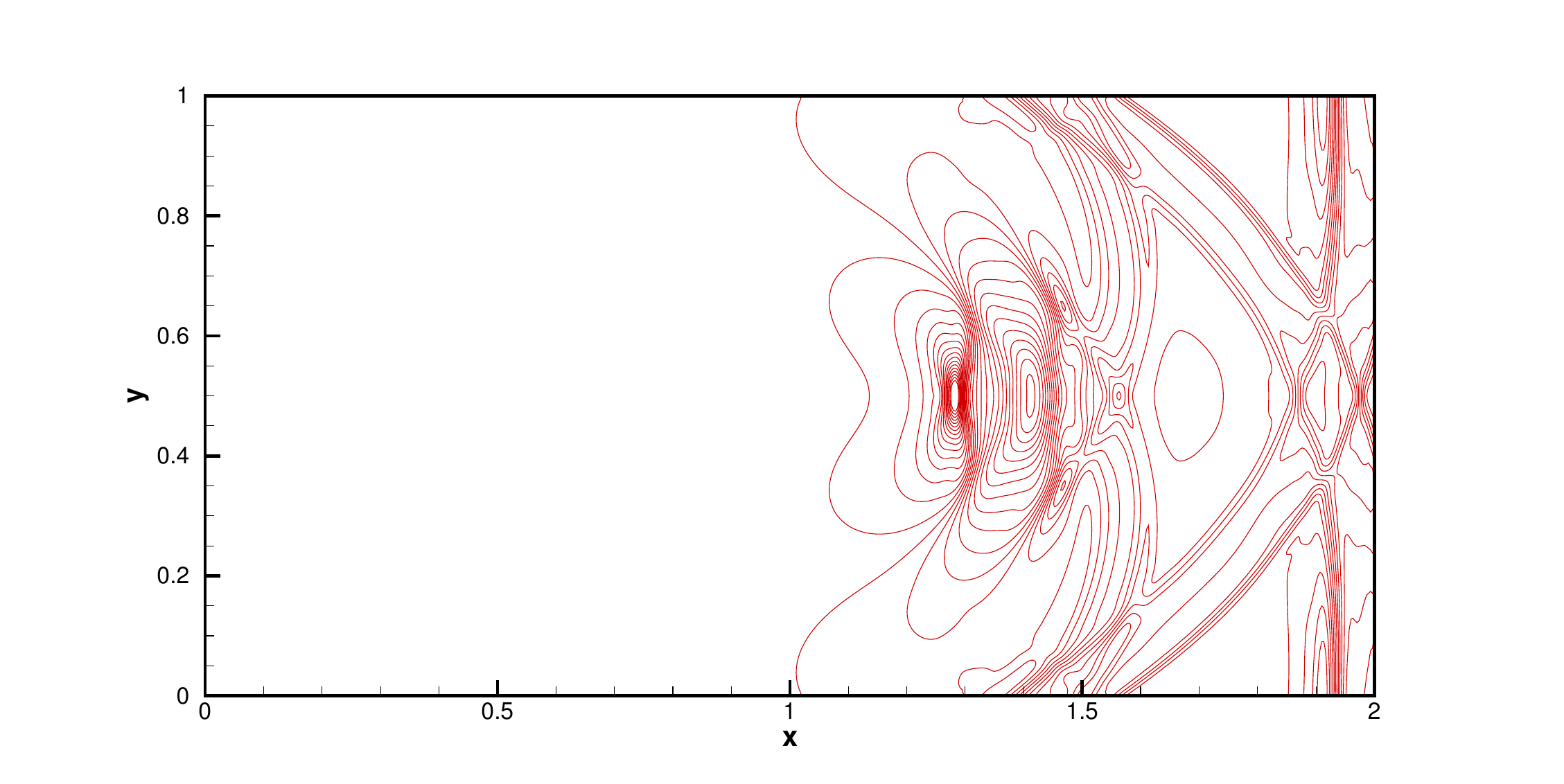}
\caption{A small perturbation of a two-dimensional steady state water flow.
Contours of water surface level h+b at different end time. From top to bottom:
at time $t = 0.12$ s from $0.999703$ to $1.00629$;
at time $t = 0.24$ s from $0.994836$ to $1.01604$;
at time $t = 0.36$ s from $0.988582$ to $1.0117$;
at time $t = 0.48$ s from $0.990344$ to $1.00497$; and
at time $t = 0.6$ s from $0.995065$ to $1.0056$.
Left: numerical results on a mesh with $200 \times 100$ uniform cells.
Right: numerical results on a mesh with $600 \times 300$ uniform cells.
 } \label{f:2D-perturbation}
\end{figure}

\end{document}